\documentclass[twoside, onecolumn, draft]{IEEEtran}
\input{symbols}
\usepackage{authblk}

\newcommand\floor[1]{\lfloor#1\rfloor}

\title{On the shape of the general error locator polynomial for cyclic codes}
\author[1]{F. Caruso}
\author[2]{E. Orsini}
\author[1]{M. Sala}
\author[1,\thanks{{\emph Corresponding authors}: maxsalacodes@gmail.com, claudia.tinnirello@gmail.com}]{C. Tinnirello}
\affil[1]{Department of Mathematics, University of Trento, Italy}
\affil[2]{Department of Computer Science, University of Bristol, United Kingdom}

\begin{document}
\maketitle
%%%%%%%%%%%%%%%%%%%%%%%%%%%%%%%%%%%%%%%%%%%%%%%%%%%%%%%%%%%%%%%%%%
%{\commentM{Non sono sicura sul titolo, il teorema 3.4
%vale per tutti i ciclici non solo i binari, forse direi
%``On the shape of the general error locator polynomial for cyclic codes and some particular binary cases''}}
\begin{abstract}
A general result on the explicit form of the general error locator polynomial for all cyclic codes is given, along
with several results for infinite classes of cyclic codes with $t=2$ and $t=3$.
From these, a theoretically justification of the sparsity of the general error locator polynomial is obtained
for all cyclic codes with $t\leq 3$ and $n<63$, except for three cases where the sparsity is proved by a computer check.
Moreover, we discuss some consequences of our results to the understanding of the complexity of bounded-distance decoding of cyclic codes.
\end{abstract}

%%%%%%%%%%%%%%%%%%%%%%%%%%%%%%%%%%%%%%%%%%%%%%%%%%%%%%%%%%%%%%%%%%%
\begin{IEEEkeywords}
cyclic codes, bounded-distance decoding, general error locator polynomial, symmetric functions, computational algebra,
finite fields, Groebner basis
\end{IEEEkeywords}
%==================================================================

%\vfill\pagebreak

\bigskip
\bigskip
\section{Introduction} \label{sec:introduction}
%\subsection*{Introduction and  previous work}
This paper focuses primarily on some issues concerning the efficiency of bounded-distance decoding for cyclic codes.
Cyclic codes form a large class of widely used error correcting codes. They include important codes such as 
the Bose-Chaudhuri-Hocquenghem (BCH) codes, quadratic residue (QR) codes and Golay codes.
In the last fifty years many efficient bounded-distance decoders have been developed for special classes, e.g. the Berlekamp-Massey
(BM) algorithm (\cite{massey1969shift}) designed for the BCH codes. Although BCH codes can be decoded efficiently, it is known that their decoding
performance degrades as the length increases (\cite{CGC-cd-art-linweldon}). Cyclic codes are not known to suffer from the same 
distance limitation, but no efficient bounded-distance decoding algorithm is known for them (up to their actual distance).

On the other hand, the BM algorithm can also be applied to some cyclic codes, provided that there are enough consecutive
known syndromes; namely, $2t$ consecutive syndromes are needed to correct a corrupted word with at most $t$ errors. Unfortunately,
for an arbitrary cyclic code the number of consecutive known syndromes is less than $2t$.
When few unknown syndromes are needed to get $2t$ consecutive syndromes, a good strategy could be to develop an efficient method to determine
expressions of unknown syndromes in terms of known syndromes. In \cite{feng1994new} Feng and Tzeng proposed a matrix method which is based on the
existence of a syndrome matrix with a particular structure. This method depends on the weight of the error pattern, so it leads to a
step-by-step decoding algorithm, and hence the error locator polynomial may not be determined in one step. In \cite{he2001decoding} He et al.
developed a modified version of the Feng-Tzeng method, and used it to determine the needed unknown syndrome and  to decode the binary QR code of
length $47$. In \cite{chang2003algebraic,truong2008algebraic,truong2005algebraic} Chang et al. presented algebraic decoders
for other binary QR codes combining the Feng-He matrix method and the BM algorithm. Another method used to yield representations of unknown
syndromes in terms of known syndromes is the Lagrange interpolation formula (LIF)\cite{chang2010algebraic}. This method has two main problems: it
can be applied only to codes generated by irreducible polynomials and its computational time grows substantially as the number of errors increases. 
The first problem was overcome by Chang et al. in \cite{chang2012multivariate}. Here the authors introduced a multivariate interpolation formula
(MVIF) over finite fields and used it to get an unknown syndrome representation method similar to that in \cite{chang2010algebraic}. They also apply
the MVIF to obtain the coefficients of the general error locator polynomial of the $[15,11,5]$ Reed-Solomon (RS) code. Later, trying to overcome the
second problem, Lee et al.\cite{lee2012new} presented a new algorithm which combines the syndrome matrix search and the modified Chinese remainder theorem (CRT).
Compared to the Lagrange interpolation method, this substantially reduces the computational time for binary cyclic codes generated by irreducible polynomials.  

Besides the unknown syndrome representation method, other approaches have been proposed to decode cyclic codes. In 1987 Elia \cite{elia1987algebraic}  proposed
a seminal efficient algebraic decoding algorithm for the Golay code of length $23$. Orsini and Sala \cite{CGC-cd-art-gelp1} introduced the general
error locator polynomial and presented an algebraic decoder which permits to determine the correctable error patterns of a cyclic code in one
step. They constructively showed that general error locator polynomials exist for any cyclic code (it is shown to exist in a Gr\"{o}bner basis of the
syndrome ideal), and gave some theoretical results on the structure of such polynomials in \cite{orsini2007general}, without the need to
actually compute a Gr\"{o}bner basis. In particular, for all binary cyclic codes with length less than $63$ and correction capability less or equal than $2$,
they provided a sparse implicit representation, and show that most of these codes may be grouped in a few classes, each allowing a theoretical
interpretation for an explicit sparse representation. In any case, direct computer computations show that the general error locator polynomial
for all these codes is actually sparse. The efficiency of the previous method depends on the sparsity of the relevant polynomial. At present, there is no theoretical proof
of the sparsity of general error locator polynomials for arbitrary cyclic codes (and no proof for sparse representations), but there is some experimental evidence
in the binary case. The proof of its sparsity in the general case would be a significant result in complexity theory, because it would imply that the
complexity of the bounded-distance decoding problem for cyclic codes (allowing unbounded preprocessing) is polynomial in the code length.

In \cite{orsini2007general} the authors also provide a structure theorem for the general error locator
polynomials of a class of binary cyclic codes. A generalization of this result is given in \cite{chang2010algebraic}. 
The low computational complexity of the general error locator polynomial for the two error-correctable cyclic codes has motivated the studies for variations on this polynomial
% Owing to the low computational complexity for the two error-correctable codes, it is practical to use the general error locator
% polynomials to decode cyclic codes up to two errors. This motivates the studies for variations on this polynomial
\cite{marcolla2012improved, lee2010unusual, lee2011weak}. 
We note that Gr\"{o}bner bases could also be used  for online decoding. In \cite{augot2009decoding} Augot et al. proposed an online Gr\"{o}bner basis decoding algorithm which consists
of computing for each received word a Gr\"{o}bner basis of the syndrome ideal with the Newton identities,  in order   to express the coefficients of the  error locator polynomial in terms
of the syndromes of the received word.

\subsection*{Our results}

In what follows, we list the main original contributions of this paper.

\begin{itemize}
\item We give a general result on the structure of the general error
locator polynomial for {\emph all} cyclic codes, which generalizes
Theorem 1 of \cite{chang2010algebraic}.
\item We provide some results on the general error locator polynomial
for several families of binary cyclic codes with $t\leq 3$, adding
theoretical evidence to the sparsity of the general error locator
polynomial for infinite classes of codes.
\item As a first direct consequence to $t=2$, we theoretically justify the sparsity of the
general error locator polynomial for all the five remaining cases
which were not classified in \cite{orsini2007general}. 
% by exhibiting a new classification of the
% binary  cyclic codes with $n<63$ and $t=2$.
%
\item As a second direct consequence to $t=3$, we classify the cyclic codes with $n<63$ and
$t=3$ according to the shape of their general error locator polynomial, justifying theoretically the results
for all cases except three. For the remaining three cases, the general error locator polynomial can be computed explicitly.

\item Finally, we provide some results on the complexity of bounded-distance decoding of some classes of cyclic codes.
Some results are conditioned to a conjecture and others hold unconditionally.
\end{itemize}

\subsection*{Paper organization}
The remainder of the paper is organized as follows. 
In Section \ref{sec:preliminaries} we review some definitions concerning cyclic codes:
we recall Cooper's philosophy, the notion of general error locator
polynomial  for cyclic codes and how this polynomial can be use to decode.
In Section \ref{sec:generalshape} we show our main result, Theorem \ref{SparseTwoSyndromes} 
which provides a general structure of the error locator polynomial for all cyclic codes.
In Section \ref{sec:specialcases} we show how the previous results can be used to obtain a sparse
representation of the general error locator polynomial for three of the five exceptional cases of Theorem 28(\cite{orsini2007general})
and we give new results on the structure of the general error locator polynomial for some {\emph infinite} classes of binary cyclic codes with $t=2$,
including the two remaining cases.
In Section \ref{sec:$t=3$} we provide a general error locator polynomials for all binary cyclic codes with $t=3$
and $n<63$. A sparse representation is theoretically justified for all cases, except three. We also give new results on the structure of the general error locator polynomial for
some {\emph infinite} classes of binary cyclic codes with $t=3$
In Section \ref{sec:conclusions}, we draw some conclusions.

%%%%%%%%%%%%%%%%%%%%%%%%%%%%%%%%%%%%%%%%%%%%%%%%%%%%%%%%%%%%%%%%%%
\section{Preliminaries}\label{sec:preliminaries}

In this section we review standard notation.
The reader is referred to \cite{CGC-cd-book-macwilliamsTOT},
\cite{CGC-cd-book-petersonI} and \cite{CGC-cd-book-handbook}
for general references on coding theory.
Throughout the paper, $n$ will denote  an odd number $n\geq 3$.
Vectors are denoted by bold lower-case letters.
\subsection{Some Algebraic Background and Notation}
Let $q=p^s$, where $p\geq 2$ is any prime and $s\geq 1$ is any positive integer. For us, $\FF_q$ denotes the finite
field with $q$ elements. 

Sometimes we will deal with rational expressions of the kind $\frac{f}{g}$, with $f,g\in \FF_q[x_1,\ldots,x_\ell]$ for some
$\ell\geq 1$. When we evaluate this expression on any point $P\in (\FF_q)^\ell$ it is possible that
$g(P)=0$. However, our rational expressions are evaluated only in points such that if $g(P)=0$ then also $f(P)=0$,
and when this happens we always use the convention that $\frac{f(P)}{g(P)}=0$.

\subsection{Cyclic codes}
A linear code $C$ is a cyclic code if it is invariant under any cyclic shift
of the coordinates. Cyclic codes  have been extensively studied in coding theory
for their useful algebraic properties.
We only consider $[n,k,d]_q$ cyclic codes with
$(n,q)=1$, that is $n$ and $q$ are coprime. Let $R=\FF_q[x]/(x^n-1)$,
each vector $\vc \in {\FF_q}^n$ is associated to a polynomial $c_0+ c_1x + \dots + c_{n-1}x^{n-1}
\in R$, and it is easy to prove that cyclic codes of length $n$ over $\FF_q$ are ideals in $R$.
Let $\FF_{q^m}$ be the splitting field of $x^n-1$  over $\FF_q$,
and let $\alpha$ be  a primitive $n$-th root of unity over $\FF_q$, then it holds
$x^n-1= \prod_{i=0}^{n-1}(x-\alpha^i)$.
Let $g(x) \in \FF_q[x]$ be the generator polynomial of an
$[n,k,d]_q$-cyclic code $C$, i.e. the monic polynomial of degree $n-k$ such that
$\langle g(x) \rangle = C$. It is well-known that $g(x)$ divides $x^n-1$ and
the set $\tilde{S}_C=\{i_1, \dots, i_{n-k}\mid g(\alpha^{i_j})=0,\; j=1,\dots,n-k\}$ is called the
{\emph{complete defining set}} of $C$. Also, the roots of unity $\{\alpha^i \mid i \in \tilde{S}_C \}$ are called the {\emph{zeros}}
of the cyclic code $C$. Notice that the complete defining set permits to specify a cyclic code. 
By this fact, we can write a parity-check matrix for $C$  as an $(n-k) \times n$ matrix $H=\{h_{j\ell}\}_{j,\ell}$
over $\F$ such that   $h_{j\ell}=\alpha^{\ell  i_j}$, where  $\alpha$ is 
a fixed primitive $n$th roots of unity, $i_j \in \tilde{S}_C$ and $\ell=0,\dots, n-1$.
As the complete defining set is partitioned into cyclotomic classes, any subset of $\tilde{S}_C$ containing 
at least one element per cyclotomic class is sufficient to specify the code. We call such a set 
a {\emph{defining set}} of $C$. We will use $S_C$ to denote a defining set which is not necessarily a complete defining set.

%%%%%%%%%%%%%%%%%%%%%%%%%%%%%%%%%%%
\subsection{Cooper's philosophy}

In this section we describe the so-called {\emph{Cooper's philosophy}} approach to decode cyclic codes
up to their true error correction capability \cite{mora2009decoding}. 
The high level idea here is to reduce the decoding problem to
that of solving a polynomial system of equations where the unknowns are the error locations and the error values.

Given an
$[n,k,d]_q$ code $C$, we recall that the {\emph error correction capability} of $C$
is $t=\lfloor(d-1)/2 \rfloor$, where   $\lfloor x \rfloor$ denotes the greatest
integer less than or equal to $x$.
Let $\vc, \vr, \ve \in (\F_q)^n$ be, respectively, the
transmitted codeword, the received vector and the error vector, then 
$\vr = \vc + \ve$. If we apply the previously indicated parity-check matrix $H$ to $\vr$, we get 
$H\vr = H (\vc + \ve) = H \ve = \vs \in (\F_{q^m})^{n-k}$. 
The vector $\vs$ is called {\emph{syndrome vector}}. Recall that a {\emph{correctable syndrome}} is a syndrome 
corresponding to an error vector $\ve$ with Hamming weight $\mu \leq t$.
If there is an   error vector $\ve$ of weight $\mu \leq t$, such that
$${\bf e}=(\underbrace{0,\dots,0}_{l_1-1},
\underset{\underset{l_1}{\uparrow}}{e_{l_1}},
0,\dots,0,\underset{\underset{l_k}{\uparrow}}{e_{l_k}}, 0,\dots,0,
\underset{\underset{l_{\mu}}{\uparrow}}{e_{l_\mu}},
\underbrace{0,\dots,0}_{n-1-l_\mu}) \,,$$
then the set $L=\{l_1, \dots, l_{\mu}\}\subset\{0,\ldots,n-1\}$ is the set of the 
{\emph{error positions}}, the set $\{\alpha^l \mid l \in L\}$
is the set of the {\emph{error locations}}, and $\{e_{l_1}, \dots, e_{l_\mu}\}$ is the set of the {\emph{error values}}.
The {\emph{classical error locator polynomial}} associated to the error $\ve$ 
is the polynomial $\sigma_e(z)= \prod_{l \in L}(1-z \alpha^l)$, i.e. the polynomial having 
as zeros the reciprocal of the error locations; whereas the {\emph{plain error locator polynomial}} is the polynomial  
$
\mathbf{L}_e(z)= \prod_{l \in L}(z-\alpha^{l})
$.

It is clear that finding $\mathbf{L}_e(z)$ (or $\sigma(z)$) is equivalent to 
finding the error $\ve$: once it is found the decoding process concludes by applying the Chien search 
\cite{CGC-cd-art-ChienSearch} to find the error positions.
Associating variables $Z=(z_1, \dots, z_t)$ to the error locations, $X=(x_1,\dots,x_{n-k})$ to the syndromes components
$s_i$, and $Y=(y_1,\dots,y_t)$ to the error values, the Cooper's philosophy approach consists of writing a polynomial
system of equations (and the corresponding solutions form the {\emph{syndrome variety}} or CHRT-{\emph{variety}}), as follows.
To accomodate for the case when the error number, $|L|$, is strictly less than $t$, it is convenient to add {\emph ghost error
locations}, that is, $t-|L|$ zero values for the $z_i$'s. As we will show in a second, Cooper's equations will remain valid
with this assumption. The syndromes are rewritten in terms of power sums functions 
$$
x_j = \sum_{h=1}^{t} y_h z_h^{i_j}, \quad 1 \leq j \leq n-k ;
$$
to specify where these values lie the equations
$$
x_j^{q^m} -  x_j, \quad z_h^{n+1} - z_h, \quad y_h^{q-1} -1,  \quad  {\mbox{ for }}  1 \leq h \leq t, \;  1 \leq j \leq n-k, 
$$
are added, and finally, using
$$
z_h  z_{h'} p_{h,h'}, \quad {\mbox{ where }} \quad 
p_{h,h'}= (z_h^n - z_{h'}^n)/(z_h- z_{h'}), 1 \leq h < h' \leq t, 
$$
it is guaranteed that the  locations (if not zero) are all distinct. 
Studying the structure of the syndrome variety and its \GR \ basis, 
it was proved  (\cite{CGC-cd-art-gelp1}) the existence of a 
{\emph{general error locator polynomial}} for every cyclic code.
In more details, a general error locator polynomial $\mathcal{L}$
for an $[n,k,d]_q$ cyclic code $C$ is a polynomial in $\F_q[X,z]$,
with $X=(x_1,\dots,x_{n-k})$ such that
\begin{itemize}
\item  $\mathcal{L}(X,z)=z^t+a_{t-1}(X)z^{t-1}+\cdots+a_0(X)$,
      with $a_j\in \FF_q[X]$,  $0 \leq j \leq t-1$;
\item given a correctable syndrome
${\bf{s}}=(\bar x_1,\dots,
\bar x_{n-k})$, 
if we evaluate the $X$ variables in ${\bf{s}}$, then the $t$ roots
of $\mathcal{L}({\bf{s}},z)$ are the $\mu$ error locations plus  zero counted
with multiplicity $t-\mu$.
\end{itemize} 
Notice  that the second property is equivalent to
$\mathcal{L}({\bf{s}},z)=z^{t-\mu} \mathbf{L}_e(z)$, where $e$ is the error associated
to syndrome $\bf{s}$.
Also, the general error locator polynomial $\mathcal{L}$ does not depend on the errors actually occurred,
but it is computed in a preprocessing fashion once and for all. As a consequence, the decoding algorithm is made up of the following steps:
\begin{itemize}
 \item Compute the syndrome vector ${\bf{s}}$ corresponding to the received vector ${\bf{r}}$;
 \item Evaluate $\mathcal{L}$ at the syndromes ${\bf{s}}$;
 \item Apply the Chien search on $\mathcal{L}(\vs,z)$ to compute the error locations $\{\alpha^l \mid l \in L\}$ ;
 \item Compute the error values $\{e_{l} \mid l \in L\}$.
 \end{itemize} 
This approach is efficient as long as the evaluation of $\mathcal{L}$ is efficient (see Section \ref{sec:complexity}).

% In \cite{orsini2007general} it was shown that  this is the case
% for every binary cyclic codes with error correction capability $t \leq 2$ and length $n <63$.

% \subsection{Some complexity considerations} \label{subsec:complexity}
% 
% It is well-known that the complexity of decoding a BCH code is polynomial in the code length, while the complexity of
% decoding a linear code is likely to be exponential. Indeed, the problem of decoding a linear code remains hard even if
% the code is known in advance and can be processed for as long as desired in order to devise an ad-hoc decoding algorithm
% (\cite{bruck1990hardness}). The introduction of the general error locator polynomial by Orsini and Sala in 2004 prompted
% Sala to conjecture the sparsity of a general error locator polynomial for any cyclic code (MEGA 2005). Several numerical
% confirmations of this conjecture came up and some of them were collected in \cite{orsini2007general}.
% We would like to point out that the sparsity of the locator has nothing to do with the effort made to compute it:
% the proof of the sparsity in the general case would still be a significant result in complexity theory, because it would imply 
% that the complexity of the problem of decoding cyclic codes (allowing unbounded preprocessing) is polynomial in the code length.
% This would not contradict the results in \cite{bruck1990hardness} since the latter deals with the even more general situation of a
% {\emph linear} code.

\subsection{General error locator polynomials for some binary  cyclic codes}

Here we recall some techniques used in \cite{orsini2007general} to efficiently compute a general error locator polynomial
for binary cyclic codes without using \Gr \ bases.   
In this section we only deal with binary cyclic codes and we will often  shorten ``binary cyclic (linear) code'' to ``code'' when
it is clear from the context.

Let $C$ be a code with error capability $t=2$, ${\bf s}$ a correctable syndrome, and
$\bar z_1$ and $\bar z_2$ the (possibly ghost) error locations corresponding to the syndrome ${\bf s}$.
Then, by definition we know that  ${\mathcal{L}}(X,z)=z^2+a z+b$, where $a,b\in \FF_2[X]$, and
$b({\bf s})=\bar z_1 \bar z_2$, $a({\bf s})=\bar z_1+\bar z_2$.
Moreover, there are exactly two errors if and only if $b({\bf s})\not= 0$, and
there is exactly one  error if and only if $b({\bf s})= 0$ and
$a({\bf s})\not=0$ (in this case the error location is $a({\bf s})$).

\begin{definition}
We denote by $\mathcal{V}^\mu$ the set of syndromes corresponding to $\mu$ errors,
with $0 \le \mu \le t$.
The set of {\em correctable syndromes}  $\mathcal{V}$ is given by the (disjoint) union
of the syndromes $\mathcal{V}^\mu$ for $\mu=0, \dots, t$ (corresponding to $0$, $\dots$,  $t$ errors, respectively),
i.e. $ \mathcal{V}= \mathcal{V}^0\sqcup \mathcal{V}^1\sqcup \dots, \sqcup \mathcal{V}^t$.
\end{definition}
\begin{definition}
Let $C$ be a code with $t=2$.
A polynomial  ${\sf h}(X)$ in $\FF_2[x_1,\ldots,x_{n-k}]=\FF_2[X]$ is called 
a {\bf bordering polynomial} if
$$
   {\sf h}({\mathcal V}^0)={\sf h}({\mathcal V}^1)=\{0\},
   \quad {\sf h}({\mathcal V}^2)=\{1\} \,.
$$
\end{definition}
%
%Bordering polynomials are important within our arguments
%in this correspondence and the main reason is
%given by the following results.
%
The importance of bordering polynomials
comes from the following facts:
\begin{proposition}[\cite{orsini2007general}]
\label{prop:b*}
Let $C$ be a code with $t=2$.
Let ${\mathcal L}^*=z^2+a(X) z+b^*(X)$ be a polynomial in
$\FF_2[x_1,\ldots,x_{n-k},z]=\FF_2[X,z]$, s.t.
${\mathcal L}^*({\bf s},z)$ is an error locator polynomial for any
weight-$2$ error corresponding to a syndrome ${\bf s} \in \mathcal{V}^2$.
%Let ${\mathcal L}'_C= z^2+a z+b$ be a  general error locator polynomial for $C$.
Let ${\mathcal L}'_C= z+a(X)$ be an error locator polynomial for any weight-$1$ error,
and  ${\sf h}$ be a bordering polynomial. Then
$$
  {\mathcal L}_C(X,z)\,=\; z^2+a(X) z+b^*(X) {\sf h}(X)
$$
is a general error locator polynomial for $C$.
\end{proposition}

\begin{remark}
Notice that if $S_C$ contains $0$ then, as shown in \cite[p. 1105]{orsini2007general}, a bordering polynomial
for $C$ is given by $(1+x_0)$ with $x_0$ the syndrome corresponding to $0$. 
\end{remark}

Another useful definition from \cite{orsini2007general} is that of 
$\spl$ code.
\begin{definition} \label{defstrict}
Let $C$ be a code, we say that $C$ is a
{\bf{strictly-two-error-correcting code}} (briefly
s2ec code) if, when we know that exactly two errors have
occurred (that is $\mu=2$), then we can correct them.
\end{definition}

Trivially every code with distance $d \geq 5$ is a s2ec code,
however there are  s2ec codes with distance $d=3$,
e.g. the code defined by $n=9$ and $S_C=\{ 1 \}$ (see \cite{orsini2007general}
for details).

In \cite{orsini2007general} a complete classification of all binary cyclic code with $t \leq 2$ and $ n\leq 63$,
with respect to their defining sets is given. 

\begin{theorem}[\cite{orsini2007general}]
\label{theostruct4}
Let $C$ be an $[n,k,d]$ code with $d \in  \{5,6\}$
and $7 \leq n <63$ ($n$ odd).
Then there are seven cases:
\begin{enumerate}
\item[$1.$] either $n$ is such that the code with defining set $\{0,1\}$ has distance
      at least~$5$,
\item[$2.$] or $C$ is a BCH code, i.e. $S_C=\{1,3\}$,
\item[$3.$] or $C$ admits a defining set of type $S_C=\{1,n-1,l\}$, with $l=0,n/3$,
\item[$4.$] or $C$ admits a defining set of type $S_C=\{1,n/l  \}$,
      for some $l\geq 3$,
\item[$5.$] or $C$ is one of the following exceptional cases
  \begin{itemize}\label{missingCases}
   \item[$a)$] $n=31$, $S_C=\{1,15\}$,
   \item[$b)$] $n=31$, $S_C=\{1,5\}$,
   \item[$c)$] $n=45$, $S_C=\{1,21\}$,
   \item[$d)$] $n=51$, $S_C=\{1,9\}$,
   \item[$e)$] $n=51$, $S_C=\{0,1,5\}$.
  \end{itemize}
\item[$6.$] or $C$ is a sub-code of  one of the codes of the above cases,
\item[$7.$] or $C$ is equivalent to one of the codes of the above cases.
\end{enumerate}
\end{theorem}

According to this classification, the authors in \cite{orsini2007general}
provide an explicit general error locator polynomial for all the codes in the first five cases. 
For the codes in cases $6.$ and $7.$, such locators can be gotten from those of the previous cases, as described in the next theorem

\begin{theorem}[\cite{orsini2007general}]
\label{th:equivsubcodes}
Let $C$,$C'$ and $C''$ be three codes with the same length $n$ and the same correction capability $t$.
Let $\mathcal{L}_C$,$\mathcal{L}_{C'}$ and $\mathcal{L}_{C''}$ denote their respective general error locator polynomials.

If $C$ is a subcode of $C'$, then we can assume $\mathcal{L}_C=\mathcal{L}_{C'}$.

If $C$ is equivalent to $C''$ via the coordinate permutation function $\phi:(\FF_2)^n\rightarrow (\FF_2)^n$, then we can decode $C$ using $\mathcal{L}_{C''}$ (via $\phi$).
\end{theorem}

It is straightforward to see that any code $C$ (up to equivalence)
covered in Theorem~\ref{theostruct4} has a general error locator polynomial of type
$\mathcal{L}(X)=z^2+x_1 z + b(X)$, where
$x_1$ is the syndrome corresponding to $1$, since $1 \in S_C$ in all the cases.
So what remains to compute is $b(X)$.  For 
all the first four cases the shape of $b(X)$ can be theoretically determined and proved to be sparse.
%,i.e. the number of monomials in $b(X)$
% is approximately less than $n^3$ (with $n$ the length of the code).
The codes in $5.$ of Theorem~\ref{theostruct4} are those for which the general error locator polynomials
are not  theoretically justified  within \cite{orsini2007general}, and
are solely obtained with a \GR\ basis calculation.

\begin{remark}\label{remark_bordering}
When the defining set is not complete, the general error locator polynomial will not contain $n-k$ $X$ variables, but a smaller number.
If the defining set is given as small as possible, then there is only one syndrome per cyclotomic class and we call such syndromes
{\bf primary syndromes}. For example in Theorem \ref{theostruct4}, the defining sets of cases $1.,2.$ and $5.$ correspond only to primary
syndromes, while for cases $3.$ and $4.$ one syndrome could be unnecessary. The reason why we keep a formally unnecessary syndrome is that
our aim is to provide a sparse description of our locators in polynomial ring $\mathbb{F}_{q}[x_1,\ldots,x_{n-k},z]$ and to do that we are obviously
authorized to use all the ring variables. 
%On the other hand, we can trivially convert our polynomial into a polynomial with only primary syndromes as variables.
From now on we reserve the letter $r$ to denote the number of syndromes we are actually working on and 
so $r$ will be at least the number of primary syndromes and at most $n-k$. In particular, we will assume
$\mathcal{V}\subset(\F_{q^m})^r$ and $X=(x_1,\ldots, x_r)$.
\end{remark}

%%%%%%%%%%%%%%%%%%%%%%%%%%%%%%%%%%%
%\subsubsection{A Newton based decoder}
%\begin{proof}
%It is enough to inspect case 1 of Theorem \ref{theostruct3}. \,\QED
%\end{proof}

%\vfill\pagebreak
%%%%%%%%%%%%%%%%%%%%%%%%%%%%%%%%%%%%%%%%%%%%%%%%%%%%%%%%%%%%%%%%%%%
\section{A general description for the locator polynomial} \label{sec:generalshape}

In this section we give a 
new general result on the structure of the
error locator polynomial for {\emph all} cyclic  codes over $\F_q$.

Let $R_n=\{\alpha^i \mid i=0,\ldots,n-1\}$. Let us denote with $T_{n,t}$ the following set (compare with \cite[p. 131]{chang2010algebraic} and Def.~13 of 
\cite{orsini2007general})
$$
T_{n,t}=\;\{(\alpha^{l_1},\ldots,\alpha^{l_\mu},0,\ldots,0) \mid 0\leq l_1<\cdots< l_\mu<n,\, 0\leq\mu\leq t\}\subset(R_n\cup\{0\})^{t}.
$$
Let C be a cyclic code over $\FF_q$, with length $n$ and correction
capability $t$, defined by
$S_C=\{i_1,\dots, i_{r} \}$ and let $x_j$ be the syndrome corresponding to $i_j$ for $j\in\{1,\dots r\}$.
The following theorem (Theorem \ref{SparseTwoSyndromes})generalizes Theorem 1 of \cite{chang2010algebraic}, which dealt with the case where 
the code could be defined by only one syndrome.
Here we provide a description of the shape of the coefficients 
of a general error locator polynomial for cyclic codes over $\FF_q$.
We recall that these coefficients are polynomials in the syndrome variables $X$.
When they are evaluated at a correctable syndrome, corresponding 
to an error of weight $\mu \leq t$, they can be expressed as the elementary
symmetric functions on the error locations $z_1, \dots, z_\mu$ and zero (with multiplicity $t-\mu$), since the latters are the roots of the locator.
By definition of elementary symmetric functions, they can then be expressed as elementary symmetric polynomials in $\mu$ variables on the $z_1,\ldots,z_\mu$.
\indent
We will need the existence of a polynomial representation for arbitrary functions from $(\FF_q)^n$ to $\FF_q$.
This is not unique and can be obtained in several ways, including multivariate interpolation ([9]).
We report a standard formulation in the following lemma.
\begin{lemma} [{ \cite[p.~26]{mullen_panario2013}}]
\label{LemmaRepresentativePol}
Let $f:\FF_q^n\rightarrow \FF_q$. Then $f$ can be represented by a polynomial in $\FF_q[x_1,\ldots,x_n]$, that is, there exists a polynomial $P\in\FF_q[x_1,\ldots,x_n]$ such that
$P(b_1,\ldots,b_n)=f(b_1,\ldots,b_n)$ for all $(b_1,\ldots,b_n)\in\FF_q^n$. In particular the polynomial
$$
\sum_{(a_1,\ldots,a_n)\in\FF_q^n}^{}f(a_1,\ldots,a_n)[1-(x_1-a_1)^{q-1}]\cdots[1-(x_n-a_n)^{q-1}].
$$
represents $f$.
\end{lemma}
\noindent The next two lemmas clarify some links between syndromes and error locations which will be essential in our proof of Theorem \ref{SparseTwoSyndromes}.
\begin{lemma}\label{LemmaSymmetricPol}
Let $\sigma \in \FF_q[y_1,\dots, y_t]$
be a symmetric function. Then there exists $a\in\FF_q[X]$ such that for $(\bar{x}_1,\ldots,\bar{x}_r)\in \mathcal{V}^\mu$
$$
a(\bar{x}_1,\ldots,\bar{x}_r)=\sigma({z}_1,\ldots,{z}_\mu,0\ldots,0)
$$
with ${z}_1,\ldots {z}_\mu$ the error locations corresponding to $\bar{x}_1\ldots,\bar{x}_r$.
\end{lemma}
\begin{proof}
We claim that the statement is obvious for elementary symmetric functions, as follows. Let $\sigma_1\ldots \sigma_t$ be the elementary symmetric functions in $\FF_q[y_1,\dots, y_t]$.
The existence of a general error locator polynomial for any cyclic code guarantees that, for any $1\leq i \leq t$, for any $\sigma_i$ there is 
$a_i\in\FF_q[x_1,\dots, x_r]$ such that for $(\bar{x}_1,\ldots,\bar{x}_r)\in \mathcal{V}^\mu$
$$
a_i(\bar{x}_1,\ldots,\bar{x}_r)=\sigma_i({z}_1,\ldots,{z}_\mu,0\ldots,0)
$$
 with ${z}_1,\ldots {z}_\mu$ the error locations corresponding to $\bar{x}_1\ldots,\bar{x}_r$.

For the more general case of any symmetric function  $\sigma\in\FF_q[y_1,\ldots,y_t]$, we need the fundamental theorem on symmetric functions, which 
shows the existence of a polynomial $H\in \FF_q[y_1\ldots,y_t]$ such that $\sigma(y_1,\ldots,y_t)=H(\sigma_1(y_1,\ldots,y_t),\ldots,\sigma_t(y_1\ldots,y_t))$.
We can define $a=H(a_1,\ldots,a_t)\in\F_q[X]$. So for $(\bar{x}_1,\ldots,\bar{x}_r)\in \mathcal{V}^\mu$ and the corresponding locations ${z}_1,\ldots {z}_\mu$, we have
\begin{equation*}
\begin{split}
&\sigma(z_1,\ldots,z_\mu,0,\ldots,0)  = H(\sigma_1(z_1,\ldots,z_\mu,0,\ldots,0),\ldots,\sigma_t(z_1,\ldots,z_\mu,0,\ldots,0))=\\
& = H(a_1(\bar{x}_1,\ldots,\bar{x}_r),\ldots,a_t(\bar{x}_1\ldots,\bar{x}_r))=a(\bar{x}_1,\ldots,\bar{x}_r)\, . \qedhere
\end{split}
\end{equation*}
\end{proof}

\begin{lemma}\label{lemma_zero_over_a_line}
 Let $h\in\FF_q[X]$ with $\deg_{x_i}h<q$ for all $i=1,\ldots,r$, and $h(\bar{x}_1,\bar{x}_2,\ldots, \bar{x}_r)=0$ 
 for all $(\bar{x}_1,\ldots, \bar{x}_r)\in (\FF_q)^{r}$ with $\bar{x}_1\neq 0$. Let $l\in\FF_q[X]$ and $g(x_2\ldots,x_r)\in\FF_q[x_2,\ldots,x_r]$ such that 
 $h(X)=x_1l(x_1,x_2,\ldots,x_r)+g(x_2,\ldots,x_r)$. Then 
 $$
 h=g(x_2,\ldots,x_r)\quad \textrm{or} \quad h =(1-x_1^{q-1})\cdot g(x_2,\ldots,x_r).
 $$
\end{lemma}
\begin{proof}
Clearly, for any $h$ the two polynomials $l$ and $g$ are uniquely determined.\\
If $h(X)\in\FF_q[x_2\ldots,x_r]$, trivially, we have that $h(X)=g(x_2,\ldots,x_r)$. So we can suppose that $h(X)\notin \FF_q[x_2\ldots,x_r]$,
and let us define the polynomial $\bar{h}(X)=(1-x_1^{q-1})\cdot g(x_2,\ldots,x_r)$.
Note that $\deg_{x_i}\bar{h}<q$ for all $i=1,\ldots,r$. We claim that $h=\bar{h}$. Since the degree w.r.t. each variable $x_i$ of both the
polynomials $h$ and $\bar{h}$ is less than $q$, to prove our claim we suffice to show that $h(\hat{X})=\bar{h}(\hat{X})$ for all 
$\hat{X}\in(\FF_q)^r$. Let us distinguish the cases $\hat{x}_1=0$ and $\hat{x}_1\neq 0$.\\
If $\hat{x}_1=0$ then $h(\hat{X})=0\cdot  l(0,\hat{x}_2,\ldots, \hat{x}_r)+g(\hat{x}_2,\ldots,\hat{x}_r)=g(\hat{x}_2,\ldots,\hat{x}_r)$ and
$\bar{h}(\hat{X})=(1-0)\cdot g(\hat{x}_2,\ldots,\hat{x}_r)=g(\hat{x}_2,\ldots,\hat{x}_r)$. So, in this case $h(\hat{X})=\bar{h}(\hat{X})$.\\
Otherwise, let $\hat{x}_1\neq 0$. By hypothesis, $h(\hat{X})=0$. On the other hand, $\bar{h}(\hat{X})=(1-1)\cdot g(\hat{x}_2,\ldots,\hat{x}_r)=0$.
So, also in this case $h(\hat{X})=\bar{h}(\hat{X})$.
\end{proof}

\begin{theorem}\label{SparseTwoSyndromes}
Let C be a cyclic code over $\FF_q$, with length $n$ and correction
capability $t$, defined by
$S_C=\{i_1,\dots, i_{r} \}$ and let $x_j$ be the syndrome corresponding to $i_j$ for $j\in\{1,\dots r\}$.
Let $\sigma \in \FF_q[y_1,\dots, y_t]$ be a symmetric homogeneous function of total degree $\delta$, with $\delta$ a multiple of $i_1$,
and let $\lambda$ be a divisor of $n$. Then there exist polynomials $a\in \FF_q[X]$, $g\in\FF_q[x_2,\ldots,x_r]$, some non-negative integers $\delta_2,\ldots,\delta_r$ and 
some univariate polynomials $F_{h_2,\dots, h_{r}} \in \mathbb{F}_q[y]$ such that for any $0\leq\mu\leq t$,
for any $(\bar{x}_1, \bar{x}_2,\dots, \bar{x}_{r})\in\mathcal{V}^\mu$ and the corresponding error locations $z_1,\ldots,z_\mu$, we have
\begin{equation}\label{defgelp}
a(\bar{x}_1, \bar{x}_2,\dots, \bar{x}_{r})=\sigma(z_1,\ldots,z_\mu,0,\ldots,0) 
\end{equation}

and
\begin{equation}\label{sparseClaim}
        a(X) = x_1^{\delta/{i_1}}\sum_{h_r=0}^{\delta_r} \left(
                  \left( \frac{x_{r}}{x_1^{i_{r}}} \right)^{h_r}
              \cdots
              \sum_{h_{2}=0}^{\delta_{2}} \left(
                  \left( \frac{x_{2}}{x_1^{i_{2}}} \right)^{h_{2}}
                     F_{h_2,\dots,h_{r}}(x_1^\lambda)  \right) \cdots \right)+ (1-x_1^{q^m-1})\cdot g(x_2,\ldots,x_r) \, ,
\end{equation}
 where $\delta_i\leq q^m-1$, $\deg F_{h_2,\dots, h_{r}} \leq (q^m-1)/\lambda$.
\begin{proof}
We observe that (\ref{defgelp}) is immediate by Lemma \ref{LemmaSymmetricPol}.
To prove (\ref{sparseClaim}) we first show the case $\bar{x}_1\neq 0$ and then the general case.\\

{\bf Case $\bar{x}_1\neq 0$} \\
Let us consider the following map $A:\{ X \in\mathcal{V}\mid x_1\neq 0 \} \rightarrow \FF_{q^m}$ defined by 
\begin{equation}\label{defA}
A(\bar{x}_1, \bar{x}_2,\dots, \bar{x}_{r})=
\frac{\sigma(z_1,\dots, z_{\mu},0,\dots,0)}{\bar{x}_1^{\delta/{i_1}}}, 
\end{equation}
where $(z_1,\dots, z_{\mu},0,\dots,0)$ is the element of $T_{n,t}$ associated to the syndrome
vector $(\bar{x}_1, \bar{x}_2,\dots, \bar{x}_{r})$.
We claim that $A$ depends only on $(\bar{x}_1^\lambda, \bar{x}_2/\bar{x}_1^{i_{2}},\ldots, \bar{x}_{r}/\bar{x}_1^{i_{r}})$.
If our claim is true, then we have that 
\begin{equation}\label{shapeA}
A(\bar{x}_1, \bar{x}_2,\dots, \bar{x}_{r})=f(\bar{x}_1^\lambda, \bar{x}_2/\bar{x}_1^{i_{2}},\ldots, \bar{x}_{r}/\bar{x}_1^{i_{r}}) 
\end{equation}
for a function $f:(\F_{q^m})^r \rightarrow \FF_{q^m}$, and so, by Lemma \ref{LemmaRepresentativePol}, we can view $f$ as a polynomial in $\FF_{q^m}[x_1,\dots, x_{r}]$.
Since $\mathcal{V} \subset (\FF_{q^m})^r$, we can also view $A$ as a (non-unique) polynomial $A(X)\in \FF_{q^m}[X]$. 
On the other hand, (\ref{defA}) and Lemma \ref{LemmaSymmetricPol} show that $A(X) x_1^{\delta/i_1}\in \FF_{q^m}[X]$ equals a polynomial
$a\in\FF_q[X]$ and so also $A(X)$ can be chosen in $\FF_q[X]$. Therefore, by (\ref{shapeA}) also $f$ can be chosen in $\FF_q[X]$.\\
Let $\delta_2=\deg_{x_2}(f),\ldots,\delta_r=\deg_{x_r}(f)$. Then, by collecting the powers of $x_r$ in $f$, we will have
$f=\sum_{h=0}^{\delta_r} x_r^h f_{h}$, for some $f_h$'s, which are polynomials in $\FF_q[x_1,\ldots,x_{r-1}]$. We observe that for any
$2\leq i\leq r-1$ we have that, for any $0\leq h\leq \delta_r$,  
$\deg_{x_i}(f_h) \leq \deg_{x_i}(f)=\delta_i$ and there is at least one $\sf{h}$ such that 
$\deg_{x_i}(f_{\sf{h}}) =\delta_i$. We can repeat this argument on all $f_h$'s by collecting powers of $x_{r-1}$ and iterate on the other
$X$ variables, $x_1$ excluded, until we obviously obtain the following formal description
\begin{equation}\label{sparseClaimf}
        f(x_1,x_2,\ldots,x_r) \quad = \quad \sum_{h_r=0}^{\delta_r} x_r^{h_r}\sum_{h_{r-1}=0}^{\delta_{r-1}}
                   x_{r-1}^{h_{r-1}}
              \cdots
              \sum_{h_{2}=0}^{\delta_{2}} 
                   x_2 ^{h_{2}}
                     F_{h_2,\dots,h_{r}}(x_1) \, ,
\end{equation}
 where any $F_{h_2,\dots,h_{r}}$ is a univariate polynomial in $\FF_q[x_1]$.\\
From (\ref{defA}), (\ref{shapeA}) and (\ref{sparseClaimf}) we directly obtain the restriction of (\ref{sparseClaim}) to the case
$x_1\neq 0$, considering that $x_j^{q^m}=x_j$ implies $(1-x_j^{q^m-1})=0$, $\delta_i\leq q^m-1$ and $\deg F_{h_1,\dots h_r} \leq (q^m-1)/\lambda$. 

We now prove our claim that gives (\ref{shapeA}).
Let us take $(\tilde{x}_1, \dots, \tilde{x}_{r})$
and $(\bar{x}_1, \dots, \bar{x}_{r})$ such that
$\tilde{x}_k^\lambda= \bar{x}_k^\lambda$,
and $\tilde{x}_j/\tilde{x}_1^{i_j} = \bar{x}_j/\bar{x}_1^{i_j}$,
for $j=2,\dots, r$.
\\
The first relation implies
\begin{equation}\label{x1rel}
\tilde{x}_1 = \beta \bar{x}_1,
\end{equation}
for some $\beta$ such that $\beta^\lambda= 1$.
\\
Substituting $\tilde{x}_1$ for $\beta \bar{x}_1$ in the second relation for $j=2,\dots, r$, we obtain
\begin{equation}\label{x2rel}
\frac{\tilde{x}_j}{\tilde{x}_1^{i_j}} = \frac{\bar{x}_j}{\bar{x}_1^{i_j}}
\implies
\frac{\tilde{x}_j}{(\beta \bar{x}_1)^{i_j}} = \frac{\bar{x}_j}{\bar{x}_1^{i_j}}
\implies \tilde{x}_j = \beta^{i_j} \bar{x}_j.
\end{equation}
\\
Suppose that $(\tilde{x}_1, \dots, \tilde{x}_{r})\in\mathcal{V}^{\mu}$
and $(\bar{x}_1, \dots, \bar{x}_{r})\in \mathcal{V}^{\mu^\prime}$, with $\mu, \mu^\prime \le t$. From (\ref{x1rel}), we get
$\tilde{y}_1 \tilde{z}_1^{i_1} + \dots + \tilde{y}_{\mu} \tilde{z}_{\mu}^{i_1} =
\beta (\bar{y}_1 \bar{z}_1^{i_1}+ \dots + \bar{y}_{\mu^\prime} \bar{z}_{\mu^\prime}^{i_1}) =
\beta \bar{y}_1 \bar{z}_1^{i_1} + \dots + \beta \bar{y}_{\mu^\prime} \bar{z}_{\mu^\prime}^{i_1}$, where $\bar{z}_i$'s and $\bar{y}_i$'s are the locations 
and the error values, respectively, associated to $(\bar{x}_1, \dots, \bar{x}_{r})$; and similarly for $\tilde{z}_i$'s and $\tilde{y}_i$'s.
Also, from (\ref{x2rel}) we get
$\tilde{y}_1 \tilde{z}_1^{i_j} + \dots
+ \tilde{y}_{\mu} \tilde{z}_{\mu}^{i_j} =
\beta^{i_j} (\bar{y}_1 \bar{z}_1^{i_j} + \dots +
             \bar{y}_{\mu^\prime} \bar{z}_{\mu^\prime}^{i_j})$,
for $j=2,\dots, r$.
\\
Let us now take
$\hat{y}_j= \bar{y}_j$ and
$\hat{z}_j= \beta \bar{z}_j$, for  $j=1,\dots, \mu^\prime$.
Since the $\hat{z}_j$ are distinct valid error locations
(i.e. $\hat{z}_j^n=1$, for $j=1,\dots, \mu^\prime$) 
we have that
their syndromes are
\begin{eqnarray*}
\hat{x}_j  & = &
\hat{y}_1 \hat{z}_1^{i_j} + \dots + \hat{y}_{\mu^\prime} \hat{z}_{\mu^\prime}^{i_j} =
\beta^{i_j} \bar{y}_1 \bar{z}_1^{i_j} + \dots
  + \beta^{i_j} \bar{y}_{\mu^\prime} \bar{z}_{\mu^\prime}^{i_j} =  \\
& = & \beta^{i_j}
( \bar{y}_1 \bar{z}_1^{i_j} + \dots + \bar{y}_{\mu^\prime} \bar{z}_{\mu^\prime}^{i_j}) =
\tilde{y}_1 \tilde{z}_1^{i_j} + \dots
  + \tilde{y}_{\mu} \tilde{z}_{\mu}^{i_j} = \tilde{x}_j, \\
& \text{for\ } & j  =1,\dots, r \, .
\end{eqnarray*}
\\
Hence $(\hat{x}_1,\dots, \hat{x}_{r})=
(\tilde{x}_1,\dots, \tilde{x}_{r})$, which implies
that their corresponding locations and values
must be the same and unique,
%since $C$ is s$\tau$ec (Lemma~\ref{s2ec_cara}).
because $\mu, \mu^\prime \le t$.
Therefore
$\mu=\mu^\prime$, 
$\{\tilde{z}_1, \dots, \tilde{z}_{\mu} \} =
\{ \beta \bar{z}_1, \dots, \beta \bar{z}_{\mu} \}$,
and
$\{\tilde{y}_1, \dots, \tilde{y}_{\mu} \} =
\{ \bar{y}_1, \dots, \bar{y}_{\mu} \}$,
from which, using the fact that $\sigma$ is a symmetric homogeneous function of degree $\delta$, we have

\begin{multline*}
A(\tilde{x}_1,\dots, \tilde{x}_{r})  = 
\frac{\sigma(\tilde{z}_1, \dots, \tilde{z}_{\mu}, 0,\ldots, 0)}
     {(\tilde{y}_1 \tilde{z}_1^{i_1}+\dots +
       \tilde{y}_{\mu} \tilde{z}_{\mu}^{i_1})^{\delta/{i_1}}} =
       \frac{\sigma(\beta \bar{z}_1,\dots, \beta \bar{z}_{\mu}, 0,\ldots, 0)}
       {(\bar{y}_1 (\beta \bar{z}_1)^{i_1} + \dots + \bar{y}_{\mu} (\beta \bar{z}_{\mu})^{i_1})^{\delta/{i_1}}} = \\
  \frac{\beta^{\delta}\sigma(\bar{z}_1,\dots, \bar{z}_{\mu}, 0,\ldots, 0)}
       {\beta^\delta (\bar{y}_1 \bar{z}_1^{i_1} + \dots + \bar{y}_{\mu} \bar{z}_{\mu}^{i_1})^{\delta/{i_1}}} = 
 \frac{\sigma(\bar{z}_1,\dots, \bar{z}_{\mu},  0,\ldots, 0)}
     {(\bar{y}_1 \bar{z}_1^{i_1} + \dots + \bar{y}_{\mu} \bar{z}_{\mu}^{i_1})^{\delta/{i_1}}} =
A(\bar{x}_1,\dots, \bar{x}_{r}) \, .
\end{multline*}
\\

{\bf General case } \\
Let us consider the map $A$ and the polynomial $a$ introduced in the case $\bar{x}_1\neq 0$. Let us extend $A$ to all points in 
$(\FF_{q^m})^r$ defining $A(\tilde{X})=\tilde{x}_1^{\delta/{i_1}} a(\tilde{x}_1,\ldots,\tilde{x}_r)$ when 
$\tilde{X}\in\FF_{q^m}\setminus\mathcal{V}$ with $\tilde{x}_1\neq 0$, and $A(\tilde{X})$ any element in $\FF_{q^m}$ when
$\tilde{X}=(0,\tilde{x_2},\ldots,\tilde{x_r})\in\FF_{q^m}$. Since $A$ is a map from $(\FF_{q^m})^r$ to $\FF_{q^m}$, it is a polynomial function and
the associated polynomial $A(X)$ with $deg_{x_i}(A(x))<q$, for all $i =1,\ldots, r$, is unique. Now, let us consider the polynomial 
$h(X)=a(X)-x_1^{\delta/{i_1}}A(X)$. Thanks to what we proved in the case $x_1\neq 0$, we have that $h(X)$ satisfies the hypothesis of 
Lemma \ref{lemma_zero_over_a_line}. Since $h(X)\notin\FF_{q^m}[x_2,\ldots,x_r]$, by Lemma \ref{lemma_zero_over_a_line}, we have that
$h(X)= (1-x_1^{q^m})\cdot g(x_2,\ldots, x_r)$ for some $g(x_2,\ldots, x_r)\in\FF_{q^m}$. 
So $a(X)=x_1^{\delta/{i_1}}A(X)+(1-x_1^{q^m})\cdot g(x_2,\ldots, x_r)$.
\end{proof}
\end{theorem}

\begin{corollary}\label{SparseCorollary}
Let C be cyclic code over $\FF_q$ as in Theorem \ref{SparseTwoSyndromes}.
Then the coefficients of the general error
locator polynomial can be written in the form
given by the previous theorem.
\end{corollary}

\begin{corollary}
Let C be a code with $t=2$ defined by $S_C=\{i_1,\dots, i_{r} \}$, with $i_1=1$, and let $\mathcal{L}= z^2+x_1 z+b$ be a general error locator polynomial for $C$.
If C is a primitive code, i.e. $n=q^m-1$, then $b=x_1^2A$ with $A\in \FF_{q}[x_2/x_1^{i_2},\dots, x_{r}/x_1^{i_{r}}]$.
\end{corollary}

\begin{proof}
 Since $t=2$, $x_1$ is zero if and only if there are no errors.
 %we suppose that the two locations are different, i.e they belong to $T_n,t$ 
 Then, applying the previous theorem to $C$, we get that $b=x_1^2A$ with $A\in \FF_{q}[x_1^n,x_2/x_1^{i_2},\dots, x_{r}/x_1^{i_{r}}]$.
 On the other hand, since C is primitive, $x_1^n$ is zero when $x_1$ is zero, and it is $1$ when $x_1$
 is not zero. So for $\mu\in\{1,2\}$, $x_1^n=1$ and $b=x_1^2\bar{A}$ with $\bar{A}=A|_{x_1^n=1}$. We claim that $b^*=x_1^2\bar{A}$ is a valid location product also for the case
 $\mu=0$, which follows from the fact that $\mu=0$ if and only if $x_1=0$. 
\end{proof}
\noindent Previous corollary basically shows that in the case $t=2$ the term of the form
$(1-x_1^{q^m-1}) g$ does not appear in the expression of the locator coefficients.

%\section{The Remaining Cases}

%As stated in (\ref{theostruct4})
%the remaining cases are
%\begin{enumerate}
%\item $n=31$, $S_C=\{1,15\}$,\label{c1}
%\item $n=31$, $S_C=\{1,5\}$\label{c2}
%\item $n=45$, $S_C=\{1,21\}$\label{c5}
%\item $n=51$, $S_C=\{0,1,5\}$\label{c3}
%\item $n=51$, $S_C=\{1,9\}$\label{c4}
%\end{enumerate}

%In \cite{orsini2007general} the authors proved
%the following simpler result,
%that is when the code $C$ is only defined by the
%syndrome corresponding to $1$.

%%%%%%%%%%%%%%%%%%%%%%%%%%%%%%%%%%%%%%%%%%%%%%%%%%%%%%%%%%%%%%%%%
%\input{SpecialCases}
%%%%%%%%%%%%%%%%%%%%%%%%%%%%%%%%%%%%%%%%%%%%%%%%%%%%%%%%%%%%%%%%%%%%%%
\section{On some classes of codes with \texorpdfstring{$t=2$}{TEXT}} \label{sec:specialcases}

In this section we provide explicit sparse representations for infinite classes of codes with $t=2$, including all the five exceptional cases of Theorem \ref{theostruct4}.
For some of these we also show the connection with Theorem \ref{SparseTwoSyndromes}.

In Section \ref{sec:preliminaries} we have defined a bordering polynomial as a polynomial ${\sf h} \in \FF_2[X]$  such that
${\sf h}(\mathcal{V}^0)={\sf h}(\mathcal{V}^1)=\{ 0 \}$  and  $\bp(\mathcal{V}^2)= \{ 1 \}.$
If $0\in S_C$, then by Remark \ref{remark_bordering} it is trivial to exhibit a bordering polynomial.
If $0\not\in S_C$, we claim that in order to find such an $\bp$  it is sufficient to find a polynomial
$\wbp \in \FF_2[X]$, expressed in terms of primary syndromes only (see Remark \ref{remark_bordering}), such that
$$
\wbp(\mathcal{V}^0)= \wbp(\mathcal{V}^1) = \{ 0 \},  \quad {\mbox{ and  }} \quad 
0 \notin \wbp (\mathcal{V}^2),
$$
and we will call it a {\bf weakly bordering polynomial}.
\begin{lemma}
Let $\wbp$ be a weakly bordering polynomial for $C$. Then $\bp=\wbp/\overline{\wbp}$ is a bordering polynomial for $C$ where
$\overline{\wbp}$ is a rewriting of $\wbp$ using non-primary syndromes. 
\end{lemma}
\begin{proof}
Let $\wbp$ be a weakly bordering polynomial for $C$. Given a primary syndrome 
$x_j$ occurring in $\wbp$, we can consider the syndrome $x_i$ such that $x_j=x_i^2$, $2i\equiv j\pmod {n}$, and take $\overline{\wbp}$ as the
rewriting of $\wbp$ obtained with the substitution $x_j\rightarrow x_i^2$. This allows us to consider $\bp=\wbp/\overline{\wbp}$, with the usual
convention $0/0=0$. It is obvious that such an $h$ is a bordering polynomial for $C$.
\end{proof}

From now on, in order to write a general error locator polynomial for 
a binary cyclic code $C$ with $t=2$,  we will be satisfied with exhibiting $b^{*}$ (as in Proposition \ref{prop:b*}) 
and $\wbp$
as above.

\noindent The following obvious lemma characterizes s2ec
codes with $1\in S_C$ and it is a direct generalization of Lemma~45 \cite{orsini2007general}.

\begin{lemma}\label{s2ec_cara}
Let  $C$ be a code with length
$n$ and $S_C=\{i_1,i_2,\dots, i_r\}$, with $i_1=1$.
The following statements are equivalent:
\begin{itemize}
\item[a)] for any $\{\tilde{z_1},\tilde{z_2}\}$ and
$\{\bar{z}_1,\bar{z}_2\}$ subsets of $R_n$ such that
$$
{\tilde{z}_1^{i_j}}+{\tilde{z}_2^{i_j}}= {\bar{z}_1^{i_j}}+{\bar{z}_2^{i_j}}
\; \; j=1,\dots, r,
$$
we have \
$\{\tilde{z_1},\tilde{z_2}\}=\{\bar{z}_1,\bar{z}_2\}$
\item[b)]  $C$  is a {\emph{\Special}}.
\end{itemize}
\end{lemma}

The following lemma is a generalization of Lemma~$43$ in \cite{orsini2007general}.

\begin{lemma}\label{s2ec_cara2}
Let $C$ be a code with $S_C=\{ 0,1,l \}$
with correction capability $t$.
Let $C'$ be a code with the same length of $C$ and $S_C'=\{ 1,l \}$.
Then the correction capability of $C$ is $t=2$ if and only if $C'$ is 
{\emph{s2ec}}.
\end{lemma}
\begin{proof}
Suppose that $t=2$. To prove that $C'$ is {\emph{s2ec}} we will use the previous lemma.
So, let $\{\tilde{z_1},\tilde{z_2}\}$ and $\{\bar{z}_1,\bar{z}_2\}$ be subsets of $R_n$ such that ${\tilde{z}_1}+{\tilde{z}_2}=
{\bar{z}_1}+{\bar{z}_2}$ and ${\tilde{z}_1^{l}}+{\tilde{z}_2^{l}}= {\bar{z}_1^{l}}+{\bar{z}_2^{l}}$, and we shall prove that
$\{\tilde{z_1},\tilde{z_2}\}=\{\bar{z}_1,\bar{z}_2\}$.
Since we are assuming that $t=2$, $C$ is obviously {\emph{s2ec}}. 
Note also that ${\tilde{z}_1^{n}}+{\tilde{z}_2^{n}}= {\bar{z}_1^{n}}+{\bar{z}_2^{n}}=0$, since $\tilde{z_1}$, $\tilde{z_2}$,
$\bar{z}_1$, $\bar{z}_2$ are elements of $R_n$.
So, applying the previous lemma to $C$, we get that $\{\tilde{z_1},\tilde{z_2}\}=\{\bar{z}_1,\bar{z}_2\}$.\\
Suppose, vice-versa, that $C'$ is {\emph{s2ec}}. $C$ is a subcode of $C'$, then also $C$ is {\emph{s2ec}}. In addition, from
the BCH bound we know that $t\ge 1$. Then we are reduced to proving that we are able to distinguish the case with one error to the case with
two errors. Let $x_0$ be the syndrome corresponding to $0$, that is, $x_0={\bar{z}_1^{n}}+{\bar{z}_2^{n}}$. When $\mu=1$ then $x_0=1+0=1\neq 0$,
while when $\mu=2$ then $x_0=1+1=0$. This proves that $t=2$.
\end{proof}

Note that in \cite{orsini2007general} the authors proved
a special case of Theorem \ref{SparseTwoSyndromes}:
let $C$ be a binary cyclic code with length $n$ and defining set  $S_C=\{ 1 \}$, and 
%Let $G$ be the reduced \GR \ basis of $J^{x_1}_{S,2}$.
$\mathcal{L}^*$  the error locator
polynomial that is able to correct two errors, such that 
$\mathcal{L}^*=z^2+x_1 z+b^*(x_1)$.
Then there exists a polynomial $A \in \FF_2[y]$, such that
$$
b^*(x_1)= A(x_1^n)x_1^2 \, .
$$
The first four exceptional cases have $t=2$
and $S_C=\{1,l \}$, so  we can apply
Theorem~\ref{SparseTwoSyndromes} and get
$b^{*}$.
As regards the fifth case, we
consider the code  $C'$, defined by $S_{C'}=\{ 1,5 \}$, instead of
$C$ defined by $\{0,1,5 \}$, and  note that
$C'$ is a s2ec (by Lemma~\ref{s2ec_cara2}).
In this way  we can get $b^{*}$ for $C'$.
Notice that  $b^{*}$ for $C$ is the same
as that for $C'$, since the
additional syndrome (corresponding to $0$) plays
merely the role of determining the
error weight.
In conclusion,  the polynomial $b^{*}$ can be determinated for three
exceptional cases by Theorem~\ref{SparseTwoSyndromes}.
At this point  we are still left with the problem
of determining $b$, which reduces to finding
$\bp$ (or $\wbp$).
In what follows we show how we can do this for 
two out of three exceptional cases of Theorem \ref{theostruct4}.

\begin{theorem}\label{lambda}
Let $C$ be a binary code with $\{ 1,l \} \subset S_C$,
such that $l=2^v+1$, $v \ge 1$,
$t=2$ and $\gcd(l-2,n)=1$.
Then we can take $\wbp= x_1^l + x_2$, where $x_1$ is the syndrome corresponding 
to $1$ and $x_2$ is the syndrome corresponding to $l$.
\begin{proof}
We see that
\begin{eqnarray*}
x_1^l + x_2 & = (z_1 + z_2)^l + z_1^l + z_2^l =
(z_1^{2^v}+ z_2^{2^v}) (z_1+z_2) + z_1^l + z_2^l = \\
& = z_1^{2^v} z_2 + z_1 z_2^{2^v} = z_1 z_2 (z_1^{2^v-1} + z_2^{2^v-1}).
\end{eqnarray*}
We know that $z_1 z_2 \neq 0$, so  we are left to prove
$(z_1^{2^v-1} + z_2^{2^v-1})\neq 0$. This holds 
because if 
$z_1^{2^v-1} = z_2^{2^v-1}$, then 
$\gcd(l-2,n)=1$ implies  $z_1=z_2$, which is impossible.
\end{proof}
\end{theorem}

We can apply Theorem~\ref{lambda}  to describe the cases
$n=31$, with $S_C=\{1,5 \}$, and   $n=51$, with $S_C=\{ 1,9 \}$, of Theorem 
\ref{theostruct4},  taking 
$l=5$ and $v=2$ and  $l=9, v=3$, respectively.
In order to get $\bp$ for the last case treated by Theorem \ref{SparseTwoSyndromes}, i.e.
$n=51$, $S_C=\{0,1,5 \}$, we use the syndrome $x_0$
corresponding to $0$ and take $\bp=x+1$
(as in Section VI of \cite{orsini2007general}).

There are still two cases left. These require a different approach.  However, it turns out that this new approach can actually solve also the previous three cases in an even more efficient way.

%\begin{lemma}\label{nonzero}
%For $\tau=2$, even powers of syndromes
%cannot be zero.

%\begin{proof}
%Let us assume that for an even integer $s$ we have
%$z_1^s+ z_2^s = 0.$
%\noindent which implies $z_1^s = z_2^s$.
%Since $z_1^n=z_2^n=1$, we have
%$z_1^{(n,s)} = z_2^{(n,s)}$ and
%have $z_1 = z_2$, which contradicts our assumption $\tau=2$.
%\end{proof}
%\end{lemma}

%\subsection{Classification of the First Three Remaining Cases}

\begin{theorem}\label{sl}
Let $C$ be a binary cyclic code with $\{1,l, s-2 l,s-l,s\} \subset S_C$,
$t=2$, $\gcd(s-2l,n)=\gcd(l,n)=1$.
Let $x_1$,$x_2$,$x_3$,$x_4$,$x_5$ be the syndromes corresponding
respectively to
$1$,$l$,$s-2l$, $s-l$, $s$.
Then
\begin{equation}
b=b^{*}=\left( \frac{x_2 x_4 + x_5}{x_3} \right)^{l^{+}}, \; \;
\mathcal{L}= z^2+x_1 z + b,
\end{equation}
\noindent where $l^{+}$ is the inverse of $l$ modulo $n$.

\end{theorem}

\begin{proof}
Let us consider $x_4 x_2$:
%\begin{equation}
\begin{eqnarray*}
x_4 x_2 & = (z_1^{s-l} + z_2^{s-l}) ( z_1^l + z_2^l) =
z_1^s + z_2^s + z_1^{s-l} z_2^l + z_1^l z_2^{s-l} = \\
& = z_1^s + z_2^s + (z_1 z_2)^l (z_1^{s-2l} + z_2^{s-2l}) =
x_5 + (b^{*})^l x_3
\end{eqnarray*}
%\end{equation}
\noindent Therefore we have
$(b^{*})^l = (x_2 x_4 + x_5)/x_3$
\noindent and by the fact that $\gcd(l,n)=1$ we have
$b^{*} = ((x_2 x_4 + x_5)/x_3)^{l^{+}}$.

For $\mu=2$ we have
$x_3= z_1^{s-2l} + z_2^{s-2l}$
which cannot be zero because $\gcd(s-2l,n)=1$.

We also have $b^{*}=b$,
because for $\mu=1$,
$x_2 x_4 + x_5 = z^l z^{s-l} + z^s = 0$.

\end{proof}

%For the first three cases we present theorems for
%infinite classes of codes that contain them and
%for which a general error locator polynomial is found.

\begin{theorem}\label{theoOneCase}
Let $C$ be a code with $\{1, (n-1)/l \} \subset S_C$,
$l$ a power of $2$, $3 \nmid n$, $l\not=(n-1)$, and $t=2$.
Let $x_1$,$x_2$,$x_3$,$x_4$ be the syndromes corresponding to $1$,$2$,
$(n-1)/l$, $n-2$. Then

%Then the general error locator polynomial is
%\begin{equation}
%\mathcal{L}=z^2+x_1 z + \frac{x_3 x_1^2+1}{x_1 {x_2}^{l+4}}
%\end{equation}
\begin{equation*}
b^{*}= \frac{x_1}{x_3^{l}}, \; \; \wbp=x_4 x_1^2+1 , \; \;
\mathcal{L}= z^2+x_1 z + b^{*}\wbp/\overline{\wbp},
\end{equation*}
where $\overline{\wbp}=x_4 x_2+1$.
\begin{proof}
For $\mu=2$ and by the fact that
$l$ is a power of $2$ we have
%\begin{equation}
%\begin{split}
\begin{eqnarray*}
x_1 x_3^l & = & (z_1+z_2) (z_1^{(n-1)/l}+z_2^{(n-1)/l})^l = \\
& = & (z_1 + z_2) (z_1^{n-1}+ z_2^{n-1}) = z_1^n + z_1 z_2^{-1} + z_2 z_1^{-1}+z_2^n = \\
& = & \frac{z_1}{z_2}+\frac{z_2}{z_1} = \frac{z_1^2+ z_2^2}{z_1 z_2}
= \frac{(z_1 + z_2)^2}{b^{*}} = x_1^2/b^{*} \\
\end{eqnarray*}
%\end{split}
%\end{equation}
\noindent from which we get $b^{*}=x_1 x_3^{-l}$.

Let us now consider $\wbp(x_1,x_2,x_3,x_4)= x_4 x_1^2 + 1$.
%As stated in Remark~\ref{lambdaInsteadOfh}
We only need to prove that
$\wbp(\mathcal{V}_1) = 0$ and
$0 \notin \wbp(\mathcal{V}_2)$.
For $\mu=1$ we have
$\wbp = z^{n-2} z^2 + 1 = z^n + 1 = 0.$

For $\mu=2$ we have
%\begin{equation}
%\begin{split}
\begin{eqnarray*}
&(z_1^{n-2}+ z_2^{n-2}) (z_1^2 + z_2^2) + 1=
z_1^{n-2} z_2^2 + z_2^{n-2} z_1^2 + 1 = \\
& \frac{z_2^2}{z_1^2} + \frac{z_1^2}{z_2^2} + 1 =
\left(  \frac{z_2}{z_1} + \frac{z_1}{z_2}\right)^2+1.
\end{eqnarray*}
%\end{split}
%\end{equation}

We are then left with the problem of proving
\begin{equation*}
\frac{z_2}{z_1}+ \frac{z_1}{z_2} \neq 1.
\end{equation*}
This is equivalent to
$\alpha + \alpha^{-1} \neq 1$ for $\alpha=\frac{z_2}{z_1}$.
By contradiction, if we multiply by $\alpha$
we get $\alpha^2+\alpha+1 = 0$, which implies
$\alpha \in  \mathbb{F}_4 \setminus \mathbb{F}_2$, from which follows
that the order of $\alpha$ must be $3$.
However, $\alpha^n=1$ (since $z_1$ and $z_2$ are error locations)
and this is
impossible because we have assumed $3 \nmid n$.
\end{proof}
\end{theorem}

\begin{theorem}\label{containing_cases_4and5}
Let $C$ be a code with $\{1,l=2^j,r=2^j-2^i, s=2^j+2^{(i+1)}\} \subset S_C$ with $i\geq 0$ and $j\geq i+2$, and let $t=2$.
Let $x_1$,$x_2$,$x_3$,$x_4$ be the syndromes corresponding
respectively to $1$, $l$, $r$, $s$. Then
\begin{equation*}
(x_2+x_1^{l-r}x_3) {b^{*}}^{l-r}= (x_1^{s}+x_4)
\end{equation*}
%\noindent where $(2^i)^{+}$ is the inverse of $2$ modulo $n$.

\end{theorem}

\begin{proof}

We claim that in $\mathbb{F}_2[z_1,z_2]$ for all $i,j\in\mathbb{N}_0$, $j\geq i+2$:
\begin{equation*}\label{key_equation}
(z_1+z_2)^{2^j+2^{(i+1)}}+z_1^{2^j+2^{(i+1)}}+z_2^{2^j+2^{(i+1)}}=
(z_1z_2)^{2^i}\left(z_1^{2^j}+z_2^{2^j}+(z_1+z_2)^{2^i}(z_1^{2^j-2^i}+z_2^{2^j-2^i})\right) 
\end{equation*}
The statement is an immediate consequence of this claim.\\

We now prove our claim. The right-hand side of the equality we want to prove is equal to
\begin{align*}
&(z_1z_2)^{2^i}\left(z_1^{2^j}+z_2^{2^j}+(z_1^{2^i}+z_2^{2^i})(z_1^{2^j-2^i}+z_2^{2^j-2^i})\right)=\\
&(z_1z_2)^{2^i}\left(z_1^{2^i}z_2^{2^j-2^i}+z_1^{2^j-2^i}z_2^{2^i}\right)=
 z_1^{2^{(i+1)}}z_2^{2^j}+z_1^{2^j}z_2^{2^{(i+1)}}
\end{align*}
On the other hand, since $(z_1+z_2)^{2^{(i+1)}+2^j}=(z_1+z_2)^{2^{(i+1)}} (z_1+z_2)^{2^{j}}=
(z_1^{2^{(i+1)}}+z_2^{2^{(i+1)}}) (z_1^{2^j}+z_2^{2^j})$, also the left-hand side is equal
to $z_1^{2^{(i+1)}}z_2^{2^j}+z_1^{2^j}z_2^{2^{(i+1)}}$.
\end{proof}

\begin{corollary}\label{corollary1_cases45}
With the same hypothesis as in Theorem \ref{containing_cases_4and5}, and the assumptions that $i=0$ and $\gcd(r-1,n)=1$ we have that 

\begin{equation*}
b=b^{*}=\left( \frac{x_1^{s}+x_4}{x_2+x_1x_3} \right), \; \;
\mathcal{L}= z^2+x_1 z + b,
\end{equation*}
\end{corollary}
\begin{proof}
By Theorem \ref{containing_cases_4and5}, we get that
\begin{equation*}
(x_2+x_1^{l-r}x_3) {b^{*}}^{l-r}= (x_1^{s}+x_4).
\end{equation*}
Let us suppose that $i=0$ and $\gcd(r-1,n)=1$. Then, the previous equality becomes 
\begin{equation}
(x_2+x_1x_3){b^{*}}=(x_1^{s}+x_4).
\end{equation}
First we show that $\mu=2$ implies $(x_2+x_1x_3)\neq 0$. 
We have that 
\begin{align*}
&x_2+x_1x_3=(z_1^{2^j}+z_2^{2^j})+(z_1+z_2)(z_1^{2^j-1}+z_2^{2^j-1})=\\
&(z_1^{2^j}+z_2^{2^j})+(z_1^{2^j}+z_2^{2^j}+z_1z_2^{2^j-1}+z_1^{2^j-1}z_2)=z_1z_2(z_1^{2^j-2}+z_2^{2^j-2}).
\end{align*}
Now, because $\mu=2$ and $\gcd(r-1,n)=\gcd(2^j-2,n)=1$, $z_1z_2(z_1^{2^j-2}+z_2^{2^j-2})\neq 0$. Then $(x_2+x_1x_3)$ is nonzero.
So, we get $ b^{*}=(x_1^{s}+x_4)/(x_2+x_1x_3)$.\\
Furthermore, since for $\mu=1$, $x_1^{s}+x_4=z^{s}+z^{s}=0$, we also conclude that $b^*=b$.\qedhere
\end{proof}

We now apply the previous results to the
exceptional cases of Theorem \ref{theostruct4}. Summing up we get 
the following.

\subsubsection*{a) Case  $\mathbf{n=31 , S_C=\{1,15\}}$}

This is a special case of those covered
by Theorem~\ref{theoOneCase} for $l=1$, i.e.
where $x_1,x_2,x_3,x_4$ correspond to
$1,2,30,29$. Hence we have 
$
b^{*}= x_1/x_2, \; \; \wbp=x_1^2 x_4+1 \,,
$
and so the general error locator polynomial for the code is 

\begin{equation}
z^2+ x_1 z + \frac{x_1}{x_3}\left(\frac{x_1^2 x_4+1}{x_2 x_4+1}\right) \, ,
\end{equation}

\subsubsection*{b) Case $\mathbf{n=31, S_C=\{1,5\}}$}

This case is a special case of
Theorem~\ref{sl} for $l=1$ and $s=10$.
Therefore we have
that the general error locator polynomial for $C$ is
\begin{equation}
z^2+ x_1 z + \left(\frac{x_1 x_4 + x_5}{x_3}\right) \, ,
\end{equation}
where $x_1,x_3,x_4,x_5$ are the syndromes
of $1,8,9,10$ ($x_2=x_1$). 
\subsubsection*{c) Case $\mathbf{n=45, S_C=\{1,21\}}$}

This case is a covered by
Theorem~\ref{sl} for $l=2$ and $s=23$,
which gives the following general error locator polynomial
\begin{equation}
z^2+ x_1 z + \left( \frac{x_2 x_4 + x_5}{x_3} \right)^{23},
\end{equation} \,
\noindent where
$x_1,x_2,x_3,x_4,x_5$ are the syndromes
of $1,2,19,21,23$.
%
%This is true because by simply applying
%Theorem~\ref{sl} for $l=2$ and $s=19$ we get
%\begin{equation}
%b^2 = \frac{x_1 x_3 + x_2}{x_4}
%\end{equation}
%\noindent and the inverse of $2$ modulo $45$ is $23$.
%
Note that the  inverse of $2$ modulo $45$ is $23$.

\subsubsection*{d) Case $\mathbf{n=51, S_C=\{1, 9\}}$}
This is a special case of those covered
by Corollary~\ref{corollary1_cases45} for $i=0$ and $j=4$ which gives the following general error locator polynomial
\begin{equation}
 z^2+ x_1 z + \left( \frac{x_1^{18} + x_{4}}{x_1(x_1^{15}+x_{3})}\right),
\end{equation}
\noindent where
$x_1,x_3,x_4$ are the syndromes
of $1,15,18$.
\subsubsection*{e) Case $\mathbf{n=51, S_C=\{0, 1, 5\}}$}

For this code we get the general error locator polynomial as follows.\\
Thanks to Theorem~\ref{containing_cases_4and5} with $i=0$ and $j=3$, we have that $x_1(x_1^7+x_3){b^{*}}= (x_1^{10}+x_4)$ where $x_1,x_3,x_4$ are the syndromes
of $1,7,10$ respectively. Now, for some locations $(z_1,z_2)$ of weight $2$, $(x_1^7+x_3)$ becomes zero. However, when $(x_1^7+x_3)=0$, it is easy to get the value of $b^{*}$.
Indeed, 
\begin{equation}\label{eq:x1^7}
x_1^7=(z_1+z_2)^7=z_1^7+z_2^7+(z_1 z_2)\left((z_1+z_2)^5+(z_1z_2)^2(z_1+z_2)\right)
\end{equation}
So, when $(x_1^7+x_3)=0$ and $\mu=2$, we obtain that $\left((z_1+z_2)^5+{b^{*}}^2(z_1+z_2)\right)=0$, which gives $b^{*}=x_1^2$.

In conclusion, if $(x_1^7+x_3)$ is nonzero, then $b^{*}=(x_1^{10}+x_4)/(x_1(x_1^7+x_3))$, otherwise, $b^{*}=x_1^2$.\\
To unify the two representations, we use the following result, which is proved in the Appendix \ref{sec:appendix}, 
\begin{equation}\label{eq:iff51}
(x_1^7+x_3)=0\quad \textrm{if and only if}\quad x_1^{51}=1.
\end{equation}
Since $(x_1^{255}+1)=(x_1^{51}+1)F(x_1^{51})$ with $F(y)=y^4+y^3+y^2+y+1$, from \eqref{eq:iff51} we get that a general error locator polynomial for $C$ is
\begin{equation}
z^2+ x_1 z +\left( \frac{x_1^{10} + x_{4}}{x_1(x_1^7+x_3)} \right)+x_1^2 \left( \frac{{F(x_1^{51})}^2}{F(x_2^{51})}\right) ,
\end{equation}
\noindent where
$x_1,x_2,x_3,x_4$ are the syndromes
of $1,2,7,10$.

In Tables \ref{table_theorem_4.4}, \ref{table_theorem_4.6} and \ref{table_corollary1_4.7} 
we list binary cyclic codes, up to equivalence and subcodes, with length less than $121$ which are covered by Theorem~\ref{lambda},
Theorem~\ref{theoOneCase} and Corollary~\ref{corollary1_cases45} respectively. 
While, Table~\ref{table_theorem_4.5} reports the codes with length $n<105$ which are covered by Theorem~\ref{sl}.

\begin{table}[ht!]
\begin{center}
\begin{scriptsize}
\caption{Binary cyclic codes with $t=2$ and length $< 121$ covered by Theorem~\ref{lambda}}
\label{table_theorem_4.4}
\begin{tabular}{|l|l|l|l|l|l|l|}
\hline
$15,\{ 1, 3\}$ & $17,\{ 1\}$ & $21,\{ 1, 3\}$ & $25,\{ 1\}$ & $27,\{ 1, 9\}$ & $31,\{ 1, 3\}$ & $31,\{ 1, 5\}$ \\ \hline
$35,\{ 1, 5\}$ & $35,\{ 1, 3\}$ & $45,\{ 1, 21\}$ & $45,\{ 1, 3\}$ & $45,\{ 1, 9\}$ & $51,\{ 1, 3\}$ & $51,\{ 1, 9\}$ \\ \hline
$55,\{ 1\}$ & $63,\{ 1, 3\}$ & $65,\{ 1\}$ & $73,\{ 1, 9\}$ & $73,\{ 1, 3\}$ & $73,\{ 1, 5\}$ & $73,\{ 1, 17\}$ \\ \hline
$75,\{ 1, 3\}$ & $77,\{ 1, 33\}$ & $81,\{ 1, 9\}$ & $85,\{ 1, 3\}$ & $85,\{ 1, 5\}$ & $85,\{ 1, 9\}$ & $91,\{ 1, 17\}$ \\ \hline
$93,\{ 1, 3\}$ & $93,\{ 1, 9\}$ & $95,\{ 1\}$ & $99,\{ 1, 33\}$ & $99,\{ 1, 9\}$ & $105,\{ 1, 3\}$ & $115,\{ 1\}$ \\ \hline
$117,\{ 1, 9\}$ & $119,\{ 1, 17\}$ & $119,\{ 1, 13\}$ & & &  &\\ \hline
\end{tabular}
\end{scriptsize}
\end{center}
\end{table}
 
\begin{table}[ht!]
\begin{center}
\begin{scriptsize}
\caption{Binary cyclic codes with $t=2$ and length $< 105$ covered by Theorem~\ref{sl}}
\label{table_theorem_4.5}

\begin{tabular}{|l|l|l|l|l|l|l|}
\hline
$9,\{ 0, 1\}$ & $15,\{ 1, 3\}$ & $15,\{ 0, 1, 7\}$ & $17,\{ 0, 1\}$ & $21,\{ 0, 1, 5\}$ & $21,\{ 1, 3\}$ & $25,\{ 1\}$ \\ \hline
$27,\{ 1, 9\}$ & $27,\{ 0, 1\}$ & $31,\{ 0, 1, 15\}$ & $31,\{ 1, 5\}$ & $31,\{ 1, 3\}$ & $33,\{ 0, 1\}$ & $35,\{ 1, 3\}$ \\ \hline
$35,\{ 1, 5\}$ & $45,\{ 0, 1, 7\}$ & $45,\{ 1, 3\}$ & $45,\{ 1, 7, 15\}$ & $45,\{ 1, 21\}$ & $45,\{ 1, 9\}$ & $51,\{ 1, 3\}$ \\ \hline
$51,\{ 1, 9, 17\}$ & $51,\{ 0, 1, 19\}$ & $51,\{ 0, 1, 5, 11\}$ & $55,\{ 1\}$ & $63,\{ 1, 3\}$ & $63,\{ 1, 5, 9\}$ & $63,\{ 1, 27,31\}$ \\ \hline
$63,\{ 1, 9, 31\}$ & $63,\{ 1, 11,27\}$ & $63,\{ 0, 1, 31\}$ & $15,\{ 1, 21, 31\}$ & $63,\{ 1, 5, 13, 21\}$ & $63,\{ 1, 5,11, 21\}$ & $63,\{ 1, 23,27,31\}$ \\ \hline
$63,\{ 1, 5, 9, 31\}$ & $63,\{ 1, 7, 13, 27\}$ & $63,\{ 1, 5,11,27\}$ & $63,\{ 1, 23,27,31\}$ & $63,\{ 1, 5,11,21,31\}$ & $65,\{ 1, 3\}$ & $65,\{ 1, 7\}$ \\ \hline
$65,\{ 0, 1\}$ & $69,\{ 0,1,5\}$ & $73,\{ 1, 3\}$ & $73,\{ 1, 17\}$ & $73,\{ 0,1, 9\}$ & $75,\{ 1, 3\}$ & $75,\{ 0, 1,7\}$ \\ \hline
$77,\{ 1, 33\}$ & $81,\{ 1, 27\}$ & $81,\{ 0, 1\}$ & $81,\{ 1, 9\}$ & $85,\{ 1, 3\}$ & $85,\{ 1, 29\}$ & $85,\{ 1, 7, 13\}$ \\ \hline
$85,\{ 1, 7, 9\}$ & $85,\{ 1, 7, 21\}$ & $85,\{ 1, 9,17\}$ & $85,\{ 1, 13, 37\}$ & $85,\{ 1, 21, 37\}$ & $85,\{ 0, 1, 21\}$ & $87,\{ 0, 1, 5\}$ \\ \hline
$91,\{ 0, 1,17\}$ & $91,\{ 1,,9, 11, 13\}$ & $93,\{ 1, 3\}$ & $93,\{ 1, 5, 9\}$ & $93,\{ 1, 5, 15\}$ & $93,\{ 1, 5,21\}$ & $93,\{ 1, 5,45\}$ \\ \hline
$93,\{ 0, 1, 23\}$ & $93,\{ 1, 11, 15\}$ & $93,\{ 0, 1, 5, 11\}$ & $95,\{ 1\}$ & $99,\{ 1, 9\}$ & $99,\{ 1, 33\}$ & $99,\{ 0, 1\}$ \\ \hline
\end{tabular}
\end{scriptsize}
\end{center}
\end{table}

\begin{table}[ht!]
\begin{center}
\begin{scriptsize}
\caption{Binary cyclic codes with $t=2$ and length $< 121$ covered by Theorem~\ref{theoOneCase}}
\label{table_theorem_4.6}
\begin{tabular}{|l|l|l|l|l|l|l|}
\hline
$17,\{ 1\}$ & $25,\{ 1\}$ & $31,\{ 1, 15\}$ & $31,\{ 1, 3\}$ & $43,\{ 1\}$ & $55,\{ 1, 3\}$ & $65,\{ 1\}$ \\ \hline
$73,\{ 1, 9\}$ & $85,\{ 1, 21\}$ & $91,\{ 1, 3\}$ & $95,\{ 1, 7\}$ & $115,\{ 1, 7\}$ & $119,\{ 1, 13\}$ &  \\ \hline
\end{tabular}
\end{scriptsize}
\end{center}
\end{table}

For the sake of space, we do not show the codes covered by Theorem~\ref{sl} with length $105\leq n <121$.
We observe that in each table we also report BCH codes.
% and codes with $S_C={1}$

\begin{table}[ht!]
\begin{center}
\begin{scriptsize}
\caption{Binary cyclic codes with $t=2$ and length $< 121$ covered by Corollary~\ref{corollary1_cases45} }
\label{table_corollary1_4.7}
\begin{tabular}{|l|l|l|l|l|l|l|}
\hline
$15,\{ 1,3\}$ & $21,\{ 1,3\}$ & $25,\{ 1\}$ & $31,\{ 1, 3\}$ & $31,\{ 1,5,15\}$ & $35,\{ 1, 3\}$ & $45,\{1,3\}$ \\ \hline
$45,\{  1, 9, 15\}$ & $51,\{ 1, 3\}$ & $51,\{ 1, 9\}$ & $55,\{ 1, 3\}$ & $63,\{ 1, 3\}$ & $65,\{ 1, 3\}$ &  \\ \hline
$73,\{  1, 3\}$ & $73,\{ 1, 5\}$ & $73,\{ 1, 17\}$ & $75,\{ 1, 3\}$ & $85,\{ 1, 3\}$ & $85,\{ 1, 9,15\}$ &  \\ \hline
$85,\{ 1,9,37\}$ & $93,\{ 1,3\}$ & $93,\{ 1,9,15\}$ & $95,\{ 1\}$ & $105,\{ 1,3\}$ & $115,\{ 1\}$ &  \\ \hline
\end{tabular}
\end{scriptsize}
\end{center}
\end{table}

%%%%%%%%%%%%%%%%%%%%%%%%%%%%%%%%%%%%%%%%%%%%%%%%%%%%%%%%%%%%%%%%%%%%%%
\section{On some classes of codes with \texorpdfstring{$t=3$}{TEXT}} \label{sec:$t=3$}
In this section we provide explicit sparse representations for some infinite classes of binary codes with correction capability $t=3$.
We also consider all binary codes with $t=3$ and $n<63$, showing that they can be regrouped in few classes and
we provide a general error locator polynomial for all these codes.
% Some results have been obtained by a direct search method. We use ``(Search)'' to indicate this fact.
In \cite{chen1969some}  Chen produced a table of the minimum distances of binary cyclic codes of length at most
$65$. This table was extended to length at most 99 by Promhouse and Tavares \cite{promhouse1978minimum}.

The following theorem lists binary cyclic codes with $t=3$ and $n<63$ up to equivalence and subcodes that we obtain with MAGMA computer algebra system \cite{bosma1997magma}. 

\begin{theorem}\label{search_$t=3$}
Let $C$ be an $[n,k,d]$ code with $d\in\{7,8\}$ and $15\leq n < 63$ ($n$ odd). Then there are only three cases.
\begin{enumerate}
 \item[1)] Either $C$ is one of the following:
 \begin{itemize}
  \item[] $n=15, S_C=\{1,3,5\}, n=21, S_C=\{1,3,5\}, S_C=\{1,3,7,9\}, S_C=\{0,1,3,7\}$;   
  \item[] $n=23, S_C=\{1\}, n=31, S_C=\{1,3,5\}, S_C=\{0,1,7,15\}$;
  \item[] $n=35, S_C=\{1,3,5\}, S_C=\{1,5,7\}, n=45, S_C=\{1,3,5\}, S_C=\{1,5,9,15\}$;
  \item[] $n=49, S_C=\{1,3\}, n=51, S_C=\{1,3,9\}, n=55, S_C=\{0,1\}$;
 \end{itemize}
 
 \item[2)] or $C$ is a subcode of one of the codes of case $1)$;
 \item[3)] or $C$ is equivalent to one of the codes of the above cases.
\end{enumerate}

\end{theorem}
Subcodes and equivalences are described in Table \ref{table:eqcodes} in the Appendix \ref{tables}. 
By Theorem 12 \cite{orsini2007general}, we need to find a general error locator polynomial only for the codes in $1)$.
For our purposes, it is convenient to regroup the codes as showed in the following theorem.

\begin{theorem}\label{Classification_$t=3$}
Let $C$ be an $[n,k,d]$ code with $d\in\{7,8\}$ and $15\leq n < 63$ ($n$ odd). Then there are six cases

\begin{enumerate}

 \item[1)] either $C$ is a BCH code, i.e. $S_C=\{1,3,5\}$, 
 \item[2)] or $C$ admits a defining set containing $\{1, i, i+1, i+2, i+3, i+4\}$ where $i$ and $i+2$ are not zero modulo $n$, 
 \item[3)] or $C$ admits a defining set containing $\{1, 3, 2^i+2^j, 2^j-2^i, 2^j-2^{i+1}\}$ with $i\geq 0$ and $j\geq i+2$,
 \item[4)] or $C$ admits a defining set containing $\{1, 3, 9\}$ and $(n,3)=1$,
 \item[5)] or $C$ is one of the following:
 
 \begin{itemize}
  \item $n=21$, $S_C=\{0,1,3,7\}$;
  \item $n=51$, $S_C=\{1,3,9\}$;
  \item $n=55$, $S_C=\{0,1\}$.
 \end{itemize}

 \item[6)] or $C$ is a subcode of one of the codes of the above cases,
 \item[7)] or $C$ is equivalent to one of the codes of the above cases.
\end{enumerate}

\end{theorem}

\begin{proof}
It is enough to inspect Case $1)$ of Theorem \ref{search_$t=3$}
\end{proof}

\begin{corollary}
Let $C$ be a code with length $n< 63$ and distance $d\in\{7,8\}$. Then $C$ is equivalent to a code $D$ s.t $1\in S_D$.
\end{corollary}
\begin{proof}
It is an immediate consequence of Theorem \ref{search_$t=3$}. 
\end{proof}

Let $C$ be a code with $t=3$, $\bf{s}$ a correctable syndrome and $\bar{z}_1$, $\bar{z}_2$, $\bar{z}_3$ the error locations.
Then $\mathcal{L}(X,z)=z^3+a z^2+b z +c$, where $a,b,c\in\mathbb{F}_2[X]$, and $a(\mathbf{s})= \bar{z}_1+\bar{z}_2+\bar{z}_3$,
$b(\mathbf{s})= \bar{z}_1\bar{z}_2+\bar{z}_1\bar{z}_3+\bar{z}_2\bar{z}_3$, $c(\mathbf{s})= \bar{z}_1\bar{z}_2\bar{z}_3$.
Moreover, there are three errors if and only if $c(\mathbf{s})\neq 0$, there are two errors if and only if 
$c(\mathbf{s})=0$ and $b(\mathbf{s})\neq 0$, and there is one error if and only if $c(\mathbf{s})=b(\mathbf{s})=0$ and 
$a(\mathbf{s})\neq 0$. Note that from the previous corollary any code with $t=3$ and $n<63$ is equivalent to a code with $1$
in the defining set. This means that for all our codes the general error locator polynomial is of the form
$$
\mathcal{L}(X,z)=z^3+x_1 z^2+b z +c,
$$
where $x_1$ is the syndrome corresponding to $1\in S_C$. So we are left with finding the coefficients $b$ and $c$. Of course, $b$ in the $t=3$ case should not be confused with $b$ in the case of $t=2$ case. 
Also, when $3\in S_C$, actually we need to find only one of the two coefficients because in this case by Newton's identities
\cite{CGC-cd-book-macwilliamsTOT} we get $c=x_1^3+x_3+x_1 b$, which involves only known syndromes,
so from one coefficient we can easily obtain the other.
In the following, $\Sigma_{l,m}$ will denote all the six terms of the type $z_i^lz_j^m$, $i,j\in\{1,2,3\}$,
and $\Sigma_{l,m,r}$ denotes all the six terms of the type $z_i^lz_j^mz_k^r$, $i,j,k\in\{1,2,3\}$.
%Just for completeness' sake

Let us consider the codes in $1)$ of Theorem \ref{Classification_$t=3$}. We have the following well-known result.
\begin{theorem}
Let $C$ be a BCH code with $t=3$. Then $\mathcal{L}(X,z)=z^3+x_1 z^2+b z +c$ with
$$
b=\frac{(x_1^2x_3+x_5)}{(x_1^3+x_3)},\quad c=\frac{(x_1^3x_3+x_1^6+x_3^2+x_1x_5)}{(x_1^3+x_3)}
$$
\end{theorem}
\begin{proof}
It enough to apply Newton's identities.
\end{proof}

The next theorem provides a general error locator polynomial for codes in $2)$ of Th. \ref{Classification_$t=3$}.

\begin{theorem}\label{pro:consecutives}
Let $C$ be a code with $t=3$ and $S_C$ containing $\{1, i, i+1, i+2, i+3, i+4\}$ where $i$ and $i+2$ are not zero modulo $n$. 
Then $\mathcal{L}(X,z)=z^3+x_1 z^2+b z +c$ with
$$
b=\frac{x_iU+x_{i+1}V}{W},\quad c=\frac{x_{i+1}U+x_{i+2}V}{W}
$$

where $U=x_{i+4}+x_1x_{i+3}$, $V=x_{i+3}+x_1x_{i+2}$ and $W=x_{i+1}^2+x_ix_{i+2}$.
\end{theorem}

\begin{proof}
Let us suppose that three errors occur, that is, $\mathbf{e}$ has weight three, and let $\mathbf{s}$ be its syndrome vector.
It is a simple computation to show, using the Newton's identities 
\[
\begin{cases}
x_{i+4}=x_1x_{i+3}+b x_{i+2}+c x_{i+1}\\ 
x_{i+3}=x_1x_{i+2}+bx_{i+1}+cx_{i}
\end{cases}
\]
that $b=\frac{x_iU+x_{i+1}V}{W}$, and $c=\frac{x_{i+1}U+x_{i+2}V}{W}$, where 
$W=x_{i+1}^2+x_ix_{i+2}=\Sigma_{i, i+2}$ which cannot be zero because $i$ and $i+2$ are not zero modulo $n$. 
Then, when $\mu=3$, $\mathcal{L}(\mathbf{s},z)$ is the error locator polynomial for $C$.

Let us show that it is actually a general error locator polynomial for $C$.  We have that 
\begin{multline*}
x_{i+1}U+x_{i+2}V= (z_1^{i+1}+z_2^{i+1}+z_3^{i+1})\left(z_1^{i+4}+z_2^{i+4}+z_3^{i+4}+
(z_1+z_2+z_3)(z_1^{i+3}+z_2^{i+3}+z_3^{i+3})\right)+\\ 
(z_1^{i+2}+z_2^{i+2}+z_3^{i+2})\left(z_1^{i+3}+z_2^{i+3}+z_3^{i+3}+
(z_1+z_2+z_3)(z_1^{i+2}+z_2^{i+2}+z_3^{i+2})\right)=\Sigma_{1,i+1,i+3},
\end{multline*}
and 
\begin{eqnarray*}
& x_iU+x_{i+1}V= (z_1^{i}+z_2^{i}+z_3^{i})\left(z_1^{i+4}+z_2^{i+4}+z_3^{i+4}+
(z_1+z_2+z_3)(z_1^{i+3}+z_2^{i+3}+z_3^{i+3})\right)+\\ & (z_1^{i+1}+z_2^{i+1}+z_3^{i+1})\left(z_1^{i+3}+z_2^{i+3}+z_3^{i+3}+
(z_1+z_2+z_3)(z_1^{i+2}+z_2^{i+2}+z_3^{i+2})\right)=\\
& =\Sigma_{1,i,i+3}+\Sigma_{1,i+1,i+2}+\Sigma_{i+1,i+3},
\end{eqnarray*}
Let us suppose that $\mu=2$. In this case, $W=z_1^iz_2^{i+2}+z_1^{i+2}z_2^{i}$, which is again different from zero. Furthermore, $x_{i+1}U+x_{i+2}V=\Sigma_{1,i+1,i+3}$ is
zero because $\mu=2$. Finally, $x_iU+x_{i+1}V$ is different from zero because $x_iU+x_{i+1}V=\Sigma_{1,i,i+3}+\Sigma_{1,i+1,i+2}+\Sigma_{i+1,i+3}$ and 
$\Sigma_{i+1,i+3}$ cannot be zero.
When $\mu=1$, $W=z_1^iz_2^{i+2}+z_1^{i+2}z_2^{i}=0$ and $x_iU+x_{i+1}V=\Sigma_{i+1,i+3}=0$. 
\end{proof}

To obtain a general error locator polynomial for codes in $3)$ of Theorem \ref{Classification_$t=3$}, we need the following
lemma.

\begin{lemma}

Let $\sigma_k=\sum_{1\leq i_1<\cdots< i_k\leq 3}z_{i_1}\cdots z_{i_k}$ be the $k$th elementary symmetric polynomial 
in the variables $z_1,z_2,z_3$ over $\mathbb{F}_2$, where $k\in\{1,2,3\}$, 
and let $x_h=\sum_{l=1}^{3}z_l^h\in\mathbb{F}_2[z_1,z_2,z_3]$ be the power sum polynomial of degree $h$, with $h\geq 0$.
Then, for $i\geq 0$ and $j\geq i+2$,
$$
x_1^{2^i+2^j}+x_{2^i+2^j}=\sigma_2^{2^i}x_{2^j-2^i}+\sigma_3^{2^i}x_{2^j-2^{i+1}}
$$

\begin{proof}
$x_1^{2^i+2^j}=(z_1+z_2+z_3)^{2^i+2^j}=(z_1+z_2+z_3)^{2^i}(z_1+z_2+z_3)^{2^j}=x_{2^i+2^j}+\Sigma_{2^i,2^j}$. On the other
hand, $\sigma_2^{2^i}x_{2^j-2^i}=(z_1 z_2 +z_1 z_3+z_2 z_3)^{2^i}(z_1^{2^j-2^i}+z_2^{2^j-2^i}+z_3^{2^j-2^i})=
\Sigma_{2^i,2^j}+\Sigma_{2^i,2^i,2^j-2^i}$ and $\sigma_3^{2^i}x_{2^j-2^{i+1}}=(z_1z_2z_3)^{2^i}
(z_1^{2^j-2^{i+1}}+z_2^{2^j-2^{i+1}}+z_3^{2^j-2^{i+1}})= \Sigma_{2^i,2^i,2^j-2^i}$. 
So $x_1^{2^i+2^j}=x_{2^i+2^j}+\sigma_2^{2^i}x_{2^j-2^i}+\sigma_3^{2^i}x_{2^j-2^{i+1}}$. 
\end{proof}

\end{lemma}
 
\begin{theorem}\label{pro:powers}
Let $C$ be a code with $t=3$ and $S_C$ containing $\{1, 3, 2^i+2^j, 2^j-2^i, 2^j-2^{i+1}\}$ with $i\geq 0$ and $j\geq i+2$.
Then $\mathcal{L}(X,z)=z^3+x_1 z^2+b z +c$ with
$$
b=\left(\frac{x_{2^j-2^{i+1}}U+V}{W}\right)^{(2^i)^+},\quad c=\left(\frac{x_{2^j-2^i}U+x_{1}^{2^i}V}{W}\right)^{(2^i)^+}
$$

where $U=(x_1^3+x_3)^{2^i}$, $V=x_1^{2^i+2^j}+x_{2^i+2^j}$, $W=x_{2^j-2^i}+x_1^{2^i}x_{2^j-2^{i+1}}$ and $(2^i)^+$ is the inverse of
$2^i$ modulo $n$.
\end{theorem}

\begin{proof}

Since the syndrome $x_1$ is a known syndrome, that is, $1\in S_C$, we have that $a=x_1$. 
From the Newton identity $c=x_1^3+x_3+x_1b$ we get that 
\begin{equation}\label{eq:newton_2^i} 
c^{2^i}=x_1^{3\times2^i}+x_3^{2^i}+x_1^{2^i}b^{2^i}
\end{equation}
On the other hand, by the previous lemma, we have that 
\begin{equation} \label{eq:previous_lemma}
x_1^{2^i+2^j}+x_{2^i+2^j}=b^{2^i}x_{2^j-2^i}+c^{2^i}x_{2^j-2^{i+1}}
\end{equation}
Taking into account \eqref{eq:newton_2^i} and \eqref{eq:previous_lemma}, a few computations lead to the equalities 
$b^{2^i}=\left(\frac{x_{2^j-2^{i+1}}U+V}{W}\right)$ and $c^{2^i}=\left(\frac{x_{2^j-2^i}U+x_{1}^{2^i}V}{W}\right)$.
Suppose that $\mu=3$. Then $W=(z_1^{2^j-2^i}+z_2^{2^j-2^i}+z_3^{2^j-2^i})+(z_1^{2^i}+z_2^{2^i}+z_3^{2^i})
(z_1^{2^j-2^{i+1}}+z_2^{2^j-2^{i+1}}+z_3^{2^j-2^{i+1}})=\Sigma_{2^i,2^j-2^{i+1}}$. Since $j$ is an integer, it is not
possible that $2^i=2^j-2^{i+1}$, then $W$ is different from zero. Also $x_{2^j-2^i}U+x_{1}^{2^i}V=\Sigma_{2^i,2^{i+1},2^j-2^i}+
\Sigma_{2^i,2^i,2^j}$ and $x_{2^j-2^{i+1}}U+V=\Sigma_{2^i,2^{i+1},2^j-2^{i+1}}+ \Sigma_{2^{i+1},2^j-2^i}$.
From the previous computations we get that if $\mu=2$ then $W\neq 0$,  $x_{2^j-2^i}U+x_{1}^{2^i}V=0$, and 
$x_{2^j-2^{i+1}}U+V\neq 0$. The last equality is because $\Sigma_{2^{i+1},2^j-2^i}\neq 0$. Furthermore, if $\mu=1$, then
$x_{2^j-2^{i+1}}U+V=0$.
\end{proof}

Finally, let us consider the codes in $4)$ of Theorem \ref{Classification_$t=3$}. 
In \cite{elia1987algebraic} Elia presents an algebraic decoding for the $(23,12,7)$ Golay code providing the error locator
polynomials for $\mu$ errors, for $\mu$ from one to three.
In \cite{lee2011weak} Lee proves that the error locator polynomial $L^{(3)}$ corresponding to three errors is actually a weak
error locator polynomial for this code. Notice that $L^{(3)}$ is a weak error locator polynomial for all cyclic codes $C$ with $t=3$,
$S_C$ containing $\{1,3,9\}$ and $(n,3)=1$.
Next theorem proves that one can obtain a general error locator polynomial for these codes by slightly modifying $L^{(3)}$. 
\begin{theorem}\label{pro:case $4$}
Let $C$ be a code with $t=3$ and $S_C$ containing $\{1, 3, 9\}$ with $(n,3)= 1$.
Then $\mathcal{L}(X,z)=z^3+x_1 z^2+b z +c$ with
$$
b=(x_1^2+D^{l^*})h,\quad c=(x_3+x_1 D^{l^*})h,
$$
where $D=\left(\frac{x_9+x_1^9}{x_3+x_1^3} \right)+(x_1^3+x_3)^2$, $h=\frac{(x_1^3+x_3)}{(x_1x_2+x_3)}$, $l=3$ and $l^*$ is the inverse of $l$ modulo $2^m-1$ with
$\mathbb{F}_{2^m}$ the splitting field of $x^n-1$ over $\mathbb{F}_2$.
\end{theorem}

\begin{proof}
Since $1\in S_C$, we have that $a=x_1$. 
From the following Newton identities 
\[
\begin{cases}
x_{9}=x_1x_{8}+b x_{7}+c x_{6}\\ 
x_{7}=x_1x_{6}+bx_{5}+cx_{4}\\
x_{5}=x_1x_{4}+bx_{3}+cx_{2}\\
x_{3}=x_1x_{2}+bx_{1}+c
\end{cases}
\]
using the equalities $x_6=x_3^2$, and $x_{2^i}=x_1^{2^i}$ for $i\ge 0$, we get

\begin{equation}
\left(\frac{x_9+x_1^9}{x_3+x_1^3} \right)+(x_1^3+x_3)^2=(b+x_1^2)^3
\end{equation}
So $b=x_1^2+D^{l^*}$. From $x_{3}=x_1x_{2}+bx_{1}+c$, we find $c=x_3+x_1 D^{l^*}$. Let us prove that
$\mathcal{L}$ is a general error locator polynomial. By Lemma 1 and Lemma 2 in \cite{lee2011weak},
it is enough to note that when there is one error $h=0$, while when there are two or three errors $h=1$.
\end{proof}
\begin{table}[ht!]
\caption{Binary cyclic codes with $t=3$ and length $< 121$ covered by Theorem~\ref{pro:consecutives}}
\begin{center}
\begin{scriptsize}
\label{table_theorem_5.5}
\begin{tabular}{|l|l|l|l|l|l|}
\hline
$15,\{ 1, 3,5\}$ & $21,\{ 1,3,5\}$ & $21,\{ 1, 5,9\}$ & $23,\{ 0, 1\}$ & $31,\{ 0,1, 7,15\}$  \\ \hline
$31,\{ 1, 3,5\}$ & $35,\{ 1, 3,5\}$ & $35,\{ 1, 5, 7\}$ & $45,\{ 1, 3, 5\}$ & $49,\{ 1, 3\}$  \\ \hline
$63,\{ 1, 3, 5\}$ & $63,\{ 1, 3,11,23,27,31\}$ & $63,\{ 1, 5,9,13,21\}$ & $63,\{ 1, 3,11,13,23\}$& $63,\{ 1,5,11,13,15\}$  \\ \hline
$63,\{ 1,15,23,31\}$ & $63,\{ 1,5,13,15,21\}$ & $63,\{ 0, 1,15,31\}$ & $63,\{ 1, 5,9,13,15\}$ & $63,\{ 1,11,13,15,23,27\}$  \\ \hline
$69,\{ 1,3,23\}$ & $69,\{ 0, 1, 3\}$ & $75,\{ 1, 3, 5\}$ & $75,\{ 1, 3,25\}$ & $77,\{ 1, 3\}$  \\ \hline
$77,\{ 1, 7,33\}$ & $85,\{ 1, 3,5\}$ & $85,\{ 1, 7,13,15,17\}$ & $85,\{ 1, 15, 29,37\}$ & $85,\{ 0, 1, 21,37\}$ \\ \hline
$89,\{ 0, 1, 3\}$ & $89,\{ 0, 1, 11\}$ & $91,\{ 1, 3\}$ & $91,\{ 1, 9,19\}$ & $91,\{ 1, 7,9,11,13\}$ \\ \hline
$93,\{ 1, 5, 17,33\}$ & $93,\{ 1, 7, 9,17\}$ & $93,\{ 1, 15,17,31,33\}$ & $93,\{ 1, 11, 23,45\}$ & $93,\{ 1,17,23,31,33\}$   \\ \hline
$93,\{ 1, 9,17,33\}$ & $93,\{ 1, 3,5\}$ & $105,\{ 1, 3, 5\}$ & $105,\{ 1, 3, ,13,25\}$ & $105,\{ 1, 5,9,17\}$  \\ \hline
$105,\{ 1, 9, 13,25\}$ & $105,\{ 1, 5, 7, 9, 11 \}$ & $105,\{ 1, 3,9,17,25\}$ & $105,\{ 1, 3,17,21,25\}$ & $105,\{ 1, 3, 11,17,45\}$  \\ \hline
$105,\{ 1, 5, 9,49\}$ & $105,\{ 1, 3,17,25,49\}$ & $105,\{ 1, 9, 11,13,15,17\}$ & $105,\{ 1, 9,13,45,49\}$ & $105,\{ 1, 9,17,25,49\}$ \\ \hline
$105,\{ 1, 9,11,13,45\}$ & $105,\{ 0, 1, 9,13\}$ & $105,\{ 1, 3, 17,35\}$ & $113,\{ 0, 1\}$ \\ \hline
$115,\{ 1, 23, 25\}$ & $117,\{ 0, 1,3\}$ & $117,\{ 0, 1,21,29\}$ & $119,\{ 1, 3\}$ &$119,\{ 1, 7,17\}$  \\ \hline
$119,\{ 1, 11, 13\}$ & & & & \\ \hline
\end{tabular}
\end{scriptsize}
\end{center}
\end{table}

\begin{table}[ht!]
\caption{Binary cyclic codes with $t=3$ and length $< 121$ covered by Theorem~\ref{pro:powers}}
\begin{center}
\begin{scriptsize}
\label{table_theorem_5.7}
\begin{tabular}{|l|l|l|l|l|l|l|}
\hline
$15,\{ 1, 3,5\}$ & $21,\{ 1,3,5\}$ & $21,\{ 1, 3,7,9\}$ & $31,\{ 1,3,5\}$ & $35,\{ 1, 3,5\}$ & $45,\{ 1, 3,5\}$ & $49,\{ 1, 3\}$ \\ \hline
$63,\{ 1, 3, 5\}$ & $75,\{ 1, 3, 5\}$ & $77,\{ 1, 3\}$ & $85,\{ 1, 3, 5\}$ & $91,\{ 1, 3\}$ & $93,\{ 1, 3,15,31,33\}$ & $93,\{ 1, 3,7,9\}$ \\ \hline
$93,\{ 1,3,5\}$ & $105,\{ 1, 3,5\}$ & $117,\{ 1,3,7\}$ & $119,\{ 1, 3\}$ & &  &  \\ \hline
\end{tabular}
\end{scriptsize}
\end{center}
\end{table}
In Tables \ref{table_theorem_5.5}, \ref{table_theorem_5.7} we list binary cyclic codes, up to equivalence and subcodes, with length less than $121$ which are covered by
Theorem~\ref{pro:consecutives} and Theorem~\ref{pro:powers} respectively. We observe that in each table we also report BCH codes.

Table \ref{table:1} shows a general error locator polynomial for each code in Case 1) of Theorem \ref{search_$t=3$} with $n<55$. 
Since the codes in Cases 2) and 3) of Theorem \ref{search_$t=3$} are equivalent or subcodes of the codes in Case 1), so (Theorem~ \ref{th:equivsubcodes})
their general error locator polynomial is the same or can be easily deduced from one of the general error locator polynomial in the table.\\
In Table \ref{table:1} the codes are grouped according to increasing lengths and are specified with defining sets containing only primary syndromes.
For each of these codes, the coefficients $b$ and $c$ of the general error locator polynomial is reported respectively in the second column and in the third column;  
The value in the fourth column explains which point of Theorem \ref{Classification_$t=3$} has been used to describe the corresponding code
family. In all cases except case $4$ and for the codes with length $n=49$ and $n=51$, $b$ and $c$ are expressed in terms of primary syndromes: if the defining set in the last
column is $S_C=\{i_1, i_2,\ldots, i_j\}$ with $i_1< i_2<\cdots< i_j$, then $x_k$ denotes the syndrome corresponding to $i_k$, for $k=1,2,\ldots,j$. 
When $0$ belongs to the defining set, it will be treated as if it were an $n$, with $n$ the length of the code.
For instance, for the code with length $n=21$ and defining set $\{0,1,3,7\}$ the syndrome corresponding to $0$ is $x_4$.

\setlength{\tabcolsep}{2pt}
{
{\centering
\begin{table}[ht!]
\caption{Binary cyclic code with $t=3$ and $n<55$}

\begin{tabular}{|c|c|c|c|c|}
%\begin{}

\hline
$n$ & $b$ & $c$ & Case &Codes\\
\hline
$15$  & $\frac{(x_1^3x_2+x_1^6+x_2^2+x_1x_3)}{(x_1^3+x_2)}$ & $\frac{(x_1^2x_2+x_3)}{(x_1^3+x_2)}$         &  1    &${\{1,3,5\}}$ \\  \hline
$21$ & $\frac{(x_1^3x_2+x_1^6+x_2^2+x_1x_3)}{(x_1^3+x_2)}$ & $\frac{(x_1^2x_2+x_3)}{(x_1^3+x_2)}$          &  1     &${\{1,3,5\}}$ \\
     \cline{2-5}
     & $\frac{x_2^2(x_1^3+x_2)+(x_1^9+x_4)}{x_3+x_1x_2^2}$ & $\frac{x_3(x_1^3+x_2)+x_1(x_1^9+x_4)}{x_3+x_1x_2^2}$          
     &  3   &${\{1,3,7,9\}}$\\ 
     \cline{2-5}
     &
     \parbox{4.9cm}{
     \scriptsize{$x_4x_1^2+x_3^3x_1^2+x_3^2x_2^3+x_3^2x_2^2x_1^3+x_3^2x_1^9+x_3x_2^3x_1^{28} + x_3x_2^2x_1^{10}+x_3x_2x_1^{13}+x_3x_1^{37}+x_2^7x_1^{44}
    +x_2^7x_1^{23}+x_2^6x_1^{47}+x_2^6x_1^5 + x_2^5x_1^{50}+x_2^4x_1^{53}+x_2^4x_1^{32}+x_2^3x_1^{56}+x_2^3x_1^{35}+x_2^2x_1^{59}+x_2^2x_1^{38}
     +x_2x_1^{41}+x_2x_1^{20}+x_1^{23}+x_1^2$}}  
     
     & \footnotesize{$x_1^3+x_2+x_1b$}
     & 5   & ${\{0,1,3,7\}}$

\\
\hline
$23$ & \parbox{4.9cm}{\scriptsize{\begin{align*}&\left(x_1^2+\left(\frac{x_9+x_1^9}{x_3+x_1^3}+(x_1^3+x_3)^2\right)^{1365}\right)\cdot \\ &\cdot \frac{(x_1^3+x_3)}{(x_1x_2+x_3)}\end{align*}}}
& \footnotesize{$(x_1^3+x_3+bx_1)$}\small{$\frac{(x_1^3+x_3)}{(x_1x_2+x_3)}$}         
&  4     & ${\{1\}}$\\
\hline

$31$ & $\frac{(x_1^3x_2+x_1^6+x_2^2+x_1x_3)}{(x_1^3+x_2)}$ & $\frac{(x_1^2x_2+x_3)}{(x_1^3+x_2)}$          &  1     & 
${\{1,3,5\}}$\\
%     \cline{2-4}
%     & $unk$          &  unk     & ${\bf{\{0,3,11,15\}}}$ \\
     \cline{2-5}
     & $\frac{x_3^8(x_4+x_1x_3^2)+x_2^4(x_3^2+x_1x_3^4)}{(x_3^{12}+x_2^8)}$ & $\frac{x_2^4(x_4+x_1x_3^2)+x_3^4(x_3^2+
     x_1x_3^4)}{(x_3^{12}+x_2^8)}$
     %{27,28,29,30,31=0}
     &  2     &  ${\{0,1,7,15\}}$\\

\hline

%ho tolto la star sopra i due codici $\{1,5,7\} $\{3,11,15\}$ erano {}^{*}. Non ricordo più a cosa facessero riferimento, ma non mi sembra venissero citati durante l'articolo

$35$ & $\frac{(x_1^3x_2+x_1^6+x_2^2+x_1x_3)}{(x_1^3+x_2)}$ & $\frac{(x_1^2x_2+x_3)}{(x_1^3+x_2)}$          &  1     & ${\{1,3,5\}}$

     \\
\cline{2-5}
     & $\frac{x_3(x_1^{256}+x_1x_2^2)+x_1^8(x_2^2+x1^{1025})}{x_1^{16}+x_3x_1^{1024}}$ 
     & $\frac{x_1^8(x_1^{256}+x_1x_2^2)+x_1^{1024}(x_2^2+x1^{1025})}{x_1^{16}+x_3x_1^{1024}}$          
     &  2     & ${\{1,5,7\}}$\\ 
\hline

$45$ & $\frac{(x_1^3x_2+x_1^6+x_2^2+x_1x_3)}{(x_1^3+x_2)}$ & $\frac{(x_1^2x_2+x_3)}{(x_1^3+x_2)}$          &  1    &
${\{1,3,5\}}$,
\\
\cline{2-5}
     &     $\frac{x_{4}(x_1x_3^2+x_1^{64})+x_1^{16}(x_3^2+x_1^{513})}{x_1^{512}x_{4}+x_1^{32}}$ &
     $\frac{x_1^{16}(x_1x_3^2+x_1^{64})+x_1^{512}(x_3^2+x_1^{513})}{x_1^{512}x_{4}+x_1^{32}}$    &  2        &
${\{1,5,9,15 \}}$\\
\hline
$49$ &   $\frac{(x_1^3x_2+x_1^6+x_2^2+x_1x_3)}{(x_1^3+x_2)}$ & $\frac{(x_1^2x_2+x_3)}{(x_1^3+x_2)}$          &  1       &
         ${\{1,3\}}$ \\
%& & & \\
%& & & \\
\hline

$51$ &    \parbox{6.2cm}{
     \small{$x_1^2+(x_1^3+x_2)(\frac{x_3^2+x_5x_3}{q_1x_1}+(\frac{x_3+x_2^3}{x_5^4+x_2^3}+1)(\frac{x_1^2}{x_1^3+x_2}+\frac{x_4+x_2^4x_1}{q_2}))$}\\
     
     \scriptsize{$q_1=(x_3x_1^9+x_3x_2x_1^6+x_2^3x_1^9+x_3^2+x_3x_2^2x_1^3+x_2^4x_1^6+x_3x_2^3+x_5x_1^3+x_5x_2+x_2^6)$\\
     
     $q_2=(x_1^{16}+x_2^4x_1^4+x_4x_2+x_2^5x_1)$
     }}
     &         \parbox{3.5cm}{\centering{\footnotesize{$x_1^3+x_2+x_1b$}}} & 5 
& ${\{1,3,9\}}$

\\
\hline

\end{tabular}

\label{table:1}

\end{table}
}
}

Codes described by the point $4$ of Theorem \ref{Classification_$t=3$} maintain the notation of Proposition \ref{pro:case $4$}, so $x_i$ denotes
the syndrome corresponding to $i$. In the case of the code with $n=49$, $x_1,x_2,x_3$ denote the syndromes corresponding to $1,3,5$ respectively, while for the codes with length
$51$, $x_1,x_2,x_3,x_4,x_5$ denote the syndromes corresponding to $1,3,9,13,15$ respectively.
% In the case of the codes with length $55$, for the sake of conciseness, both $b$ and $c$ are represented in the form described in Theorem \ref{SparseTwoSyndromes}, where $y_1$ stands for $x_1^{55}$.
The coefficient $a$ of the general error locator polynomial is not reported in Table \ref{table:1} because any code in Case 1) of Theorem \ref{search_$t=3$} has $1$ in its
defining set, so in all cases $a=x_1$.
A general error locator for the codes with $t=3$ and $n=55$ is showed in Table \ref{table:n=55} in the Appendix \ref{tables}.

%%%%%%%%%%%%%%%%%%%%%%%%%%%%%%%%%%%%%%%%%%%%%%%%%%%%%%%%%%%%%%%%%%%%%
%\newpage
\section{On the complexity of decoding cyclic codes} \label{sec:complexity}
In this section we estimate the complexity of the decoding approach presented in Section \ref{sec:preliminaries} for any cyclic code, along with a comparison with similar
approaches for the case where the generator polynomial of the cyclic code is irreducible.
\subsection{Complexity of the proposed decoding approach}
\begin{definition}\label{peso}
Let $\KK$ be any field and let $f$ be any (possibly multivariate) polynomial with coefficients in $\KK$, that is, 
$f\in\KK[\texttt{a}_1,\ldots,\texttt{a}_N]$ for a variable set $A=\{\texttt{a}_1,\ldots,\texttt{a}_N\}$.
We will denote by $|f|$ the number of terms (monomials) of $f$.
\end{definition}
\begin{definition}\label{density}
Let $A=\{\texttt{a}_1,\ldots,\texttt{a}_N\}$ and $B=\{\texttt{b}_1,\ldots,\texttt{b}_M\}$ be two variable sets. Let $\KK$ be a field and let $\mathcal{F}$ be a rational function
in $\KK(A)$. Let $F\in\KK[B]$, $f_1,\ldots,f_M\in\KK[A]$ and $g_1,\ldots, g_M\in\KK[A]$. We say that the triple $(F,\{f_1\ldots,f_M\},\{g_1\ldots,g_M\})$ is a
\textbf{rational representation} of $\mathcal{F}$ if 
$$
\mathcal{F}=F(f_1/g_1,\ldots, f_M/g_M)$$.
We say that the number 
$$|F|+\sum_{i=1}^{M}  (|f_i|-1) +\sum_{j=1,g_j \notin \KK}^{M}  |g_j|$$
is the \textbf{density} of the rational representation $(F,\{f_1\ldots,f_M\},\{g_1\ldots,g_M\})$.\\
Then, we define the \textbf{functional density} of $\mathcal{F}$, $||\mathcal{F}||$,  as the minimum among the densities of all rational representations of $\mathcal{F}$.
\end{definition}
With the notation of Definition \ref{peso} and \ref{density}, we have the following result, that shows
their  interlink and how natural Definition \ref{density} is.
\begin{theorem}
Let $A=\{\texttt{a}_1,\ldots,\texttt{a}_N\}$. If $\mathcal{F}$ is a polynomial in $\KK[A]$, rather than a rational function in $\KK(A)$, then
$$
  ||\mathcal{F}||\leq |\mathcal{F}| \,.
$$
Moreover, if $\mathcal{F}=\texttt{a}_1+\texttt{a}_2$ then $||\mathcal{F}||= |\mathcal{F}|=2$.
\end{theorem}
\begin{proof}
Let $\mathcal{F}\in \KK[A]$ and let $\rho=|\mathcal{F}|$. Then $\mathcal{F}=\sum_{i=1}^{\rho} h_i$, where
any $h_i$ is a monomial for $1\leq i\leq \rho$.\\
Let us consider the following rational representation for $\mathcal{F}$
$$
  B=\{\texttt{b}_1,\ldots,\texttt{b}_\rho\}\,, \quad F=\sum_{i=1}^{\rho} \texttt{b}_i\,, \quad
   f_i=h_i\,,\quad g_i=1,\, 1\leq i\leq \rho \,,
$$
\noindent
then the rational density of $(F,\{f_1,\ldots,f_\rho\},\{g_1,\ldots,g_\rho\})$ is
$$
  |F|+\sum_{i=1}^{\rho}  (|f_i|-1) +\sum_{j=1,g_j \notin \KK}^{\rho}  |g_j| \quad = \rho + 0 + 0 = \rho \,,
$$
\indent
which implies $||\mathcal{F}||\leq \rho$, as claimed.

To prove the case $\mathcal{F}=\texttt{a}_1+\texttt{a}_2$, we argument by contradiction assuming $||\mathcal{F}||=1$.
Let us consider the rational representation of $\mathcal{F}$ providing 
$$
  ||\mathcal{F}|| \quad = \quad |F|+\sum_{i=1}^{M}  (|f_i|-1) +\sum_{j=1,g_j \notin \KK}^{M}  |g_j| 
  \quad = \; 1 \,.
$$
\noindent
Since $|F|\geq 1$, we must have $\;|F|=1$,  $\;\sum_{i=1}^{M}  (|f_i|-1)=0$ and $\;\sum_{j=1,g_j \notin \KK}^{M}  |g_j|=0$.\\
Therefore, $M=1$, $|f_1|=1$ and $g_1=\nu\in \KK$. From $M=1$ and $|F|=1$ we have $F=\lambda \texttt{b}_1^\mu$ for
$\lambda\in \KK$ and $\mu\geq 1$, and so $\mathcal{F}=F(f_1/g_1)=(\frac{f_1}{\nu})^\mu$.
Recalling that $\mathcal{F}=\texttt{a}_1+\texttt{a}_2$, we finally have a contradiction
$$
  |f_1|=1 \implies \big|\Big(\frac{f_1}{\nu}\Big)^\mu\big|=1 \,, \qquad 
  \mbox{but} \quad |\mathcal{F}|=|\texttt{a}_1+\texttt{a}_2|=2 \,.
$$
\end{proof}
For example, the locator $\mathcal{L}\in\FF_2[z,x_1, x_2, x_3,x_4, x_5]$ for the case treated in Theorem \ref{sl} can be easily 
shown to have functional density
$||\mathcal{L}||\leq 6$, thanks to the following rational representation
$$ 
\mathcal{L}=F(f_1/g_1, f_2/g_2, f_3/g_3),
$$
where $F\in\FF_2[\texttt{b}_1,\texttt{b}_2, \texttt{b}_3]$, $\;f_1,f_2,f_3,g_1,g_2,g_3\in \FF_2[z,x_1, x_2, x_3,x_4, x_5]\;$ and
$$
  F=b_1^2+b_1 b_2+b_3^{l^{+}}, \quad f_1=z, g_1=1,\, f_2=x_1, g_2=1,\; f_3=x_2x_4+x_5,g_3=x_3 \,.
$$   
\begin{conjecture}[Sala, MEGA2005]\label{Sala2005} \ \\
Let $p\geq 2$ be a prime, $m\geq 1$ a positive integer and let $q=p^m$.
There is an integer $\epsilon=\epsilon(q)$  such that 
for any cyclic code $C$ over the field $\FF_q$ with $\,n\geq q^4-1$, 
$\;\gcd(n,q)=1$, $\;3\leq d\leq n-1$,\\
$C$ admits a general error locator polynomial $\mathcal{L}_c$ whose functional density is bounded by
$$
||\mathcal{L}_c|| \leq n^\epsilon \,.
$$
Moreover, for binary codes we have $\epsilon=3$, that is, $\epsilon(2)=3$.
\end{conjecture}

Let $C$ be a cyclic code over $\FF_q$ of length $n$. Let $d$ be its distance, $t$ its correction capability and 
$S_C=\{i_1,\ldots,i_r\}$ a defining set of $C$.
Let $\mathcal{L}_C$ be a general error locator polynomial of $C$. 
\begin{definition}\label{sparse}
If $\mathcal{L}_C\in \FF_2[x_1,\ldots,x_r]$, then we say that $\mathcal{L}_C$ is \textbf{sparse} if 
$||\mathcal{L}_C|| \leq n^3$.\\
If Conjecture \ref{Sala2005} holds and $\mathcal{L}_C\in \FF_q[x_1,\ldots,x_r]$, then we say
that $\mathcal{L}_C$ is \textbf{sparse} if 
$||\mathcal{L}_C|| \leq n^\epsilon$.\\
\end{definition}
\noindent The decoding procedure developed by Orsini and Sala in \cite{orsini2007general} consists of four steps:
\begin{enumerate}
 \item Computation of the $r$ syndromes $s_1,\ldots,s_r$ corresponding to the received vector; 
 \item Evaluation of $\mathcal{L}_C(x_1,\ldots,x_r,z)$ at $\bf{s}=(s_1,\ldots,s_r)$;  
 \item Computation of the roots of $\mathcal{L}_C(\bf{s},z)$;
 \item Computation of the error values $e_{l_1},\ldots, e_{l_\mu}$
\end{enumerate}
By analyzing the above decoding algorithm, we observe that the main computational cost is the evaluation of the polynomial $\mathcal{L}_C(x_1,\ldots,x_r,z)$ at ${\bf{s}}$, which reduces to the evaluation of its $z$-coefficients. 
Indeed, the computation of the $r$ syndromes $s_1,\ldots,s_r$ and of the roots of $\mathcal{L}_C(\mathbf{s},z)$ cost, respectively, $O(t\sqrt{n})$ and
$\max({O(t\sqrt{n}), O(t \log(\log(t))\log(n))})$ (\cite{schipani2011decoding}), while the computation of the error values using Forney's algorithm costs
$O(t^2)$ (\cite{hong1995simple}). Therefore, we can bound the total cost of steps 1, 3 and 4 with $O(n^2)$.\\
The following theorem is then clear and should be compared with the results in \cite{bruck1990hardness}, which suggest
that for linear codes an extension of Conjecture \ref{Sala2005} is very unlikely to hold.
\begin{theorem}\label{th-conjecture}
Let us consider all cyclic codes over the same field $\FF_q$ with $\;\gcd(n,q)=1$ and $\;d\geq 3$.\\
If Conjecture \ref{Sala2005} holds, they can be decoded in polynomial time in $n$, once a preprocessing has produced
sparse general error locator polynomials.
\end{theorem}
\begin{proof}
The only special situations not tackled by  Conjecture \ref{Sala2005} are the finite cases when $n<q^4-1$,
which of course do not influence the asymptotic complexity, and the degenerate case when $d=n$,
which can be decoded in polynomial time without using the general error locator algorithm.
\end{proof}
Although all reported experiments (at least in the binary case) confirm Conjecture \ref{Sala2005}, we are far from having 
a formal proof of it,
therefore we pass to estimate the cost of the crucial step 2 starting from results claimed in this paper or 
found elsewhere in the literature. \\
To estimate the cost of evaluating the polynomial $\mathcal{L}_C$ (at the syndrome vector 
$\mathbf{s}$) we will mainly use Corollary~\ref{SparseCorollary},
its consequences for the case $\lambda=n$ (which we can always choose), and the corresponding degree bound. 
We can neglect the cost of computing the values $\frac{x_h}{x_1^{i_h}}$ and consider polynomials in the new obvious variables.
In \cite{ballico2013evaluation}, Ballico, Elia and Sala describe a method to evaluate a polynomial in $\FF_q[x_1,\ldots,x_r]$ of degree $\delta$ with a complexity $O(\delta^ {r/2})$.
To estimate our $\delta$, we observe that, by Corollary~\ref{SparseCorollary}, we have a bound on the degree of each $z$-coefficient of $\mathcal{L}_C$ in any new variable and so its total degree is at most 
$$
   \delta \quad \leq \quad  \left((q^m-1)(r-1)+\frac{q^{m}-1}{n}\right) \,,
$$
\noindent
then,
using the method in \cite{ballico2013evaluation}, the evaluation of (the $z$-coefficients of) $\mathcal{L}_c$ at $\bf{s}$ costs
\begin{equation}\label{eq:eval_cost}
O\left(t \left((q^m-1)(r-1)+\frac{q^{m}-1}{n}\right)^{r/2}\right).
\end{equation}
So, we get that the cost of the decoding approach we are proposing is given by
\begin{equation}\label{eq:cost}
O\left(n^2+t \left((q^m-1)(r-1)+\frac{q^{m}-1}{n}\right)^{r/2}\right).
\end{equation}
We are going to show that there are infinite families of codes for which this approach is competitive with more straightforward methods (even for low values of $t$). 

Let us fix the number of syndromes $r$, and let $\gamma$ be an integer $\gamma\geq 1$. Let $\mathcal{C}_{r,\gamma}^q$ be the set of all codes over $\FF_q$ with length $n$  
such that the splitting field of $x^n-1$ over $\FF_q$ is $q^m-1=O(n^\gamma)$ (and $\gcd(n,q)=1$). 
For codes in $\mathcal{C}_{r,\gamma}^q$, the complexity \eqref{eq:cost}
of this decoding depends on $r$ and it is
\begin{equation}\label{eq:costC}
r\geq 2, \quad O\left(t n^{\gamma r/2}\right),\qquad \qquad
r=1, \quad O(n^2 + t n^{\frac{\gamma-1}{2}}) \,.
\end{equation}
\noindent
So, any family $\mathcal{C}_{r,\gamma}^q$ provides a class containing infinite codes which can be decoded in polynomial time, with infinite values of distance and length. Obviously, these classes extend widely the classes which are known to be decodable
in polynomial time up to the {\em{actual distance}}.

Theorem ~\ref{lambda},\ref{sl},\ref{theoOneCase}, Corollary~\ref{corollary1_cases45} and Theorem~\ref{pro:consecutives}, \ref{pro:powers}, \ref{pro:case $4$} show cases where
the previous estimation can be drastically improved, at least for $t=2$ and $t=3$. Indeed, these theorems provide (infinite) classes of codes with $t=2$ and $t=3$ for which the evaluation of
$\mathcal{L}_C$ costs $O(1)$, and so the decoding process costs $O(n^2)$. For $t=2$ and $t=3$ exhaustive searching method
cost, respectively, $O(n^2)$ and $O(n^3)$. For $t=2$ we match the best-known complexity and for $t=3$
our method is better.

\subsection{Comparison with other approaches}

In the last years, several methods were proposed for decoding {\emph{binary}} quadratic residue (QR) codes generated by irreducible polynomials.
In \cite{chang2010algebraic}, Chang and Lee propose three algebraic decoding algorithms based on Lagrange Interpolation Formula (LIF) for these codes.
They introduce a variation for the general error locator polynomial, which we may call fixed-weight locator. A {\emph fixed-weight locator} is a polynomial able to correct
all errors of a fixed weight via the evaluation of the corresponding syndromes. They develop a method to obtain a representation of the primary unknown syndrome
in terms of the primary known syndrome and a representation of the coefficients of both fixed-weight locator and general error locator polynomial for these codes. These polynomials
are explicitly obtained for the $(17,9,5)$, $(23,12,7)$, $(41,21,9)$ QR codes. In Table \ref{tab:numberofterms1} we treat these three codes one per column showing the number of terms relevant to the alternative representations. 
For each code, the second row deals with representation of the chosen primary unknown syndrome, while the last  
deal with two locators.
\begin{table}[ht!]
 \caption{Number of terms of unknown syndrome, fixed-weight locator and general error locator}
 \begin{center}
 \begin{small}
 \begin{tabular}{c|c|c|c|}
   \cline{2-4} & $(17,9,5)$ & $(23,12,7)$ & $(41,21,9)$ \\ \hline
   \multicolumn{1}{|c|}{Splitting field} & $\mathbb{F}_{2^8}$ & $\mathbb{F}_{2^{11}}$ & $\mathbb{F}_{2^{20}}$ \\ \hline
   \multicolumn{1}{|c|}{Unknown syndrome} & $5$ & $17$ & $1355$ \\
   \hline
   \multicolumn{1}{|c|}{Fixed-weight locator} & $4$ & $15$ & $1270$ \\
   \hline
   \multicolumn{1}{|c|}{General error locator} & $4$ & $76$ & $1380$ \\ \hline
 \end{tabular}
 \label{tab:numberofterms1}
 \end{small}
 \end{center}
 \end{table}

\noindent Note that, for all the three codes, the general error locator polynomials are sparse 
(even without using the rational representation) as forseen in Conjecture \ref{Sala2005}.
In particular the $(41,21,9)$ code has correction capability $t=4$ and the number of
terms of its locator is less than $n^\epsilon={41}^3=68921$. Observe also that the evaluation of the locators of the $(23,12,7)$ code in Table \ref{table:1} and in \cite{chang2010algebraic} cost
approximately the same.

In \cite{lee2010more}, Chang et al. propose to decode {\emph{binary}} cyclic codes generated by irreducible polynomials using, as in \cite{chang2010algebraic}, an interpolation formula in order to get 
the general error locator polynomial but in a slightly different way. The general error locators they obtain satisfy 
at least one congruence relation, and they are explicitly found for the $(17,9,5)$ QR code, the $(23,12,7)$ Golay code, 
and one $(43,29,6)$ cyclic code. 
Table \ref{tab:numberofterms2} shows the maximum number of terms for the coefficients of these three polynomials.
\begin{table}[ht!]
 \caption{Maximum number of terms among the locator coefficients $\sigma_i$}
 \begin{center}
 \begin{small}
 \begin{tabular}{c|c|c|c|}
   \cline{2-4} & $(17,9,5)$ & $(23,12,7)$ & $(43,29,6)$ \\ \hline
   \multicolumn{1}{|c|}{Splitting field} & $\mathbb{F}_{2^8}$ & $\mathbb{F}_{2^{11}}$ & $\mathbb{F}_{2^{14}}$ \\ \hline
   \multicolumn{1}{|c|}{General error locator} & $9$ & $203$ & $25$ \\ \hline
 \end{tabular}
 \label{tab:numberofterms2}
 \end{small}
 \end{center}
 \end{table}
\noindent Also in this case, the locators are {\emph{sparse}} for the three codes.
In \cite{lee2012algebraic}, Lee et al. extend the method proposed by Chang and Lee in \cite{chang2010algebraic} for finding fixed-weight locators and general error locators
for binary cyclic codes generated by irreducible polynomials to the case of {\em ternary cyclic codes} generated by irreducible polynomials.
These polynomials are presented for two {\em ternary} cyclic codes, one $(11,6,5)$ code and one $(23,12,8)$ code. 
In Table \ref{tab:numberofterms3} we report the maximum number of
terms of the coefficients of the general error locator for these two codes.

\begin{table}[ht!]
 \caption{Maximum number of terms among the locator coefficients $\sigma_i$}
 \begin{center}
 \begin{small}
 \begin{tabular}{c|c|c|}
   \cline{2-3} & $(11,6,5)$ & $(23,12,8)$ \\ \hline
   \multicolumn{1}{|c|}{Splitting field} & $\mathbb{F}_{3^5}$ & $\mathbb{F}_{3^{11}}$ \\ \hline
   \multicolumn{1}{|c|}{General error locator} & $232$ & $15204$ \\ \hline
 \end{tabular}
 \label{tab:numberofterms3}
 \end{small}
 \end{center}
 \end{table}
\noindent To discuss the {\em sparsity} of these cases one would need to know $\epsilon(3)$ from Conjecture \ref{Sala2005}.
Assuming an optimistic stance, let us compare their sparsity with $\epsilon(3)=3$, that is, let us assume the polynomial
exponent of the ternary codes to be the same as that of binary codes 
(reasonably $\epsilon(3)\geq \epsilon(2)$).\\
The first locator is definitely sparse, with $|\mathcal{L}|=232 < 1331=11^3$. For the second locator
we have $|\mathcal{L}|=15204$ which compared to  $n^3=23^3=12167$ show that the locator is not sparse
(although the numbers are close) and indeed we believe much sparser locators exist for this code, still to be found.
% In Table we summarize the comparison of our method with the previous two methods based on Lagrange interpolation formula We can see that 

In the same paper (\cite{lee2012algebraic}) the authors give also an interesting upper bound on $|\mathcal{L}|$ which
holds for any irreducible ternary cyclic code, as follows.
% for the two locators.
%
\begin{proposition}[\cite{lee2012algebraic}]\label{pro:lee}
Let $C$ be a ternary cyclic code of length $n$ with defining set $S_C=\{1\}$, and error correction capability $t$. 
% Each coefficient of a fixed-weight error locator polynomial for
% $C$ can be expressed as a polynomial in terms of the known syndrome $x_1$. The number of terms of this polynomial is less or equal than $\floor{\frac{\sum_{\nu=1}^{t} 2^\nu \binom{n}{\nu}}{n}}$. 
% Similarly,
Each coefficient of a general error locator polynomial can be expressed as a polynomial in terms of the known syndrome $x_1$ and the number of terms of this polynomial
is less than $\floor{\frac{\sum_{\nu=1}^{t} 2^\nu \binom{n}{\nu}}{n}}$.
\end{proposition}
Indeed, we can generalize their result to the following theorem holding over any finite field.
\begin{theorem}\label{ex-ternary}
Let $C$ be any cyclic code over $\FF_q$ of length $n$ with defining set $S_C=\{1\}$, $\gcd(n,q)=1$ and 
error correction capability $t$. 
% Each coefficient of a fixed-weight error locator polynomial for
% $C$ can be expressed as a polynomial in terms of the known syndrome $x_1$. The number of terms of this polynomial is less or equal than $\floor{\frac{\sum_{\nu=1}^{t} 2^\nu \binom{n}{\nu}}{n}}$. 
% Similarly,
Each coefficient of a general error locator polynomial can be expressed as a polynomial in terms of the known syndrome 
$x_1$ and the number of terms of this polynomial is less than $\floor{\frac{\sum_{\nu=1}^{t} (q-1)^\nu \binom{n}{\nu}}{n}}$.
\end{theorem}
\begin{proof}
By considering Corollary~\ref{SparseCorollary} and the fact that to obtain any locator coefficient, one can use
simply (univariate) Lagrange interpolation on the set of correctable syndromes, which are obviously
$1+\sum_{\nu=1}^{t} (q-1)^\nu \binom{n}{\nu}$.
\end{proof}
 % 
% In \cite{lee2012new}, Chang et al. suggest to predetermine a unified representation of the unknown syndrome as rational function in terms of known syndromes of a binary cyclic code
% combining the syndrome matrix search and a modified Chinese remainder theorem. 
% For binary cyclic generated by irreducible polynomials compared to the method based on LIF this procedure seems to reduce the computational time. They list in a table possible 
% choices for syndromes matrices to obtain the unknown syndrome representation for some binary cyclic codes of length less than or equal to $51$. 
% They do not report the number of monomials of all the representations they get but they say that all the obtained representations are sparse. 
% In \cite{chang2012more}, Chang et al. propose two algebraic decoders for the $(17,9,5)$ cyclic code and the $(43,29,6)$ cyclic code respectively. In both decoders they propose to
% interpolate extra points such that more binary polynomials for the unknown syndrome are generated. For the $(17,9,5)$ code, they get four different representations of the unknown
% syndrome $x_3$. The number of monomials of these representations goes from $5$ to $9$; whereas for the $(43,29,6)$ code, they get ten different representations of the unknown
% syndrome $x_3$ whose number of monomials goes from $7$ to $17$.

With $q$ fixed, the codes covered by the previous theorem are actually the component of
 our families $\mathcal{C}_{1,\gamma}^q$ for $\gamma\geq 1$. Depending on the actual considered length we will
have the correct determination of $\gamma$, since this value strongly depends on the size of the splitting field.
By \eqref{eq:cost} case $r=1$, the time complexity of the decoding method for codes in $\mathcal{C}_{1,\gamma}^q$ 
%(using the estimation for the number of the terms of the locator given by Corollary \ref{SparseCorollary}) 
is 
\begin{equation}\label{eq:costC1}
O\left(n^2+t n^{(\gamma-1)/2}\right).
\end{equation}
Using the estimation given by Proposition \ref{pro:lee}, the complexity of the same decoding approach for these codes is
\begin{equation}\label{eq:costC1Lee}
O\left(n^2+t n^{t-1}\right).
\end{equation}
We observe that which of the two estimations is better depends on the particular values of $t$ and $\gamma$.

%%%%%%%%%%%%%%%%%%%%%%%%%%%%%%%%%%%%%%%%%%%%%%%%%%%%%%%%%%%%%%%%%%%%%
\section{Conclusions} \label{sec:conclusions}
This paper provides additional theoretical arguments supporting the sparsity of the general error locator polynomial for infinite families of cyclic codes over $\FF_q$.
For infinite classes of binary codes with $t=2$ and $t=3$ a sparse general error locator polynomial is obtained. Furthermore, 
for all binary cyclic codes with length less than $63$ and correction capability $3$, we see that the number of monomials never exceeds five times the code length.

We provide some argument showing the link between the locators' sparsity and the bounded-distance decoding complexity of cyclic codes, which might turn out to be of interest.
% Also, overwhelming experimental evidence is given for the codes whose locator remains unjustified theoretically.
% In particular, for all binary cyclic codes with length less than $63$ and correction capability $3$, we see that the number of monomials never exceeds five times the code length.
%%%%%%%%%%%%%%%%%%%%%%%%%%%%%%%%%%%%%%%%%%%%%%%%%%%%%%%%%%%%%%%%%%%%%
\appendices
\section{Some tables}\label{tables}
Table \ref{table:eqcodes} report the codes with $t=3$ and $n<63$ grouped according to increasing lengths, and, within the same length according to Theorem \ref{search_$t=3$},
i.e. if two codes with the same length are equivalent or one is a subcode of the other,then they are in the same group.  
For each group there is a code in bold, which is the one reported in Table \ref{table:1}, i.e. the code for which we determined a general error locator polynomial
and that can be used to obtain locators for all the codes of the group.
\begin{table*}[ht!]
\begin{center}
 \begin{small}
\caption{Binary cyclic codes with $t=3$ and $n<63$}
\begin{tabular}{|c|>{\centering\let\newline\\\arraybackslash\hspace{0pt}}m{13cm} |}
\hline
$n$ & Codes\\
\hline

$15$ & \small{ ${\bf{\{1,3,5\}}}$, $\{3,5,7\}$, $\{ 0, 3, 5, 7 \}$, $\{ 0, 1, 3, 5 \}$} \\ \hline
$21$ & \small{ ${\bf{\{1,3,5\}}}$, $\{1,5,9\}$, $\{ 1, 3, 5, 9 \}$} \\  \cline{2-2}
           
     &\small{ ${\bf{\{1,3,7,9\}}}$, $\{3,5,7,9\}$, $\{ 0, 3, 5, 7, 9 \}$, $\{0,1,3,7,9 \}$}\\ 
     \cline{2-2}
      & \small{${\bf{\{0,1,3,7\}}}$, $\{0,5,7,9\}$, $\{ 0, 1, 3, 7, 9 \}$, $\{0,3,5,7,9 \}$}\\ \hline
$23$ & \small{${\bf{\{1\}}}$, $\{5\}$, $\{0,1 \}$, $\{0,5 \}$} \\ \hline

$31$ & \small{${\bf{\{1,3,5\}}}$,  $\{1,5,7\}$, $\{3,5,15\}$, $\{3,11,15\}$, $\{ 0,1, 5, 7 \}$, $\{ 0, 3, 5, 15 \}$, $\{ 0, 1, 3, 5 \}$, $\{1,3,11\}$, 
$\{1,7,11\}$, $\{ 0, 1, 7, 11 \}$, $\{ 0, 1, 3, 11 \}$, $\{5,7,15 \}$, $\{ 0, 5, 7, 15 \}$, $\{7,11,15 \}$, $\{ 0,3,11,15 \}$} \\

     \cline{2-2}
     &  \small{${\bf{\{0,1,7,15\}}}$, $\{0,1,3,15\}$, $\{0,3,7,11\}$, $\{0,5,11,15\}$, $\{0,1,5,11\}$, $\{0,3,5,7\}$}\\ \hline

$35$ & \small{${\bf{\{1,3,5\}}}$, $\{1,3,15\}$, $\{ 1, 3, 5, 15 \}$} \\ \cline{2-2}
     &  \small{${\bf{\{1,5,7\}}}$, $\{3,7,15\}$   
    $\{ 0, 3, 5, 7, 15 \}$,
    $\{ 0, 1, 5, 7, 15 \}$,    $\{ 0, 3, 7, 15 \}$, 
    $\{ 0, 1, 5, 7 \}$,
 $\{ 3, 5, 7, 15 \}$, $\{ 1, 5, 7, 15 \}$}\\ 
\hline

$45$ & \small{${\bf{\{1,3,5\}}}$,
$\{1,5,21 \}$, $\{5,7,21\}$, $\{3,5,7\}$,   
$\{ 0, 1, 3, 5, 9, 21 \}$,
$\{ 3, 5, 7, 9, 15 \}$,
$\{ 1, 3, 5, 9, 21 \}$, 
$\{ 1, 5, 9, 21 \}$, 
$\{ 1, 5, 15, 21 \}$,
$\{ 0, 1, 5, 9, 21 \}$, 
$\{ 0, 1, 5, 15, 21 \}$,
$\{ 0, 3, 5, 7, 9, 15 \}$, 
$\{ 1,  3, 5, 9 \}$, 
$\{ 0, 3, 5, 7, 15, 21 \}$,
$\{ 0, 1, 3, 5, 9 \}$,
$\{ 3, 5, 7, 15, 21 \}$, 
%    { 1, 3, 5 },
$\{ 0, 3, 5, 7, 21 \}$,
$\{ 1, 5, 9, 15, 21 \}$,
$\{ 3, 5, 7, 21 \}$, 
$\{ 0, 3, 5, 7, 9, 15, 21 \}$,  
$\{ 0, 1, 3, 5 \}$,
$\{ 5, 7, 9, 21 \}$, 
$\{ 1, 3, 5, 9, 15 \}$,
$\{ 0, 5, 7, 9, 21 \}$, 
$\{ 0, 1, 5, 9, 15, 21 \}$, 
%    { 0, 5, 7, 9, 15 },
$\{ 0, 1, 3, 5, 9, 15 \}$,
$\{ 0, 1, 3, 5, 15, 21 \}$,
$\{ 3, 5, 7, 9, 15, 21 \}$,
%    { 5, 7, 9, 15 },
$\{ 1, 3, 5, 15, 21 \}$,
$\{ 3, 5, 7, 15 \}$,
$\{ 0, 3, 5, 7, 15 \}$,$\{ 0, 1, 3, 5, 9, 15, 21 \}$, 
$\{ 5, 7, 15, 21 \}$,
$\{ 1, 3, 5, 9, 15, 21 \}$,$\{ 0, 5, 7, 15, 21 \}$,
$\{ 0, 3, 5, 7, 9, 21 \}$, 
$\{ 0, 1, 3, 5, 21 \}$,
$\{ 1, 3, 5, 15 \}$,
$\{ 3, 5, 7, 9, 21 \}$,
$\{ 5, 7, 9, 15, 21 \}$,
$\{ 3, 5, 7, 9 \}$,
%    { 1, 5, 21 },
$\{ 0, 1, 3, 5, 15 \}$,
$\{ 1, 3, 5, 21 \}$,
$\{ 0, 5, 7, 9, 15, 21 \}$,
$\{ 0, 1, 5, 21 \}$,
$\{ 0, 3, 5, 7, 9 \}$,

%    { 3, 5, 7 },
$\{ 0, 3, 5, 7 \}$,
$\{ 5, 7, 21 \}$,
%    { 0, 1, 5, 9, 15 },
%    { 1, 5, 9, 15 },
$\{ 0, 5, 7, 21 \}$
}
\\
\cline{2-2}
     &
${\bf{\{1,5,9,15 \}}}$, $\{5,7,9,15 \},$
$\{3,5,7,9,15 \}$, $\{0,3,5,7,9,15 \}$, $\{1,5,9,15,21 \}$,
$\{0,3,5,7,9,15,21 \}$, $\{1,3,5,9,15 \}$, $\{0,1,5,9,15,21 \}$,$\{0,5,7,9,15 \}$, $\{0,1,3,5,9,15 \}$, $\{3,5,7,9,15,21 \}$,
$\{0,1,3,5,9,15,21 \}$, $\{1,3,5,9,15,21 \}$, $\{0,5,7,9,15,21 \}$, 
$\{0,1,5,9,15 \}$ \\
\hline
$49$ & ${\bf{\{1,3\}}}$ \\

\hline

$51$ & \small{${\bf{\{1,3,9\}}}$,
$\{3,9,11 \}$, $\{3,9,19 \}$, $\{3,5,9\}$ 
$\{ 1, 3, 9, 17 \}$,
$\{ 0, 1, 3, 9 \}$,
$\{ 3, 5, 9, 17 \}$,
$\{ 0, 3, 5, 9, 17 \}$,
$\{ 3, 9, 11, 17 \}$,
$\{ 0, 3, 9, 11 \}$,
$\{ 3, 9, 17, 19 \}$,
$\{ 0, 3, 9, 11, 17 \}$,
$\{ 0, 3, 9, 17, 19 \}$,
$\{ 0, 3, 9, 19 \}$,
$\{ 0, 1, 3, 9, 17 \}$,
$\{ 0, 3, 5, 9 \}$
}
\\
\hline

$55$ & \small{ ${\bf{\{0,1\}}}$,$\{0,3 \}$}
\\
\hline

\end{tabular}
\label{table:eqcodes}
\end{small}
\end{center}
\end{table*}

In Table \ref{table:n=55} we show the coefficients $b$ and $c$ of a general error locator polynomial for binary cyclic codes with $t=3$ and $n=55$. 
For the sake of conciseness, both $b$ and $c$ are represented in the form described in Theorem \ref{SparseTwoSyndromes}, where $y_1$ stands for $x_1^{55}$.

\begin{tiny}
\begin{table*}[ht!]
\begin{center}
\caption{General error locator for cyclic codes with $t=3$ and $n=55$}
\begin{tabular}{|l|l|}
%\begin{}

\hline

b & \parbox{15.5cm}{
     \scriptsize{${x_1^2}\cdot \bigl(y_1^{475}+y_1^{472}+y_1^{470}+y_1^{469}+y_1^{468}+y_1^{463}+y_1^{462}+y_1^{461}+y_1^{460}+y_1^{458}+y_1^{457}+y_1^{455}+y_1^{454}+y_1^{452}+y_1^{449}+y_1^{
448}+y_1^{446}+y_1^{444}+y_1^{443}+y_1^{440}+y_1^{436}+y_1^{434}+y_1^{427}+y_1^{426}+y_1^{425}+y_1^{424}+y_1^{417}+y_1^{416}+y_1^{413}+y_1^{410}+y_1^{408}+y_1^{
405}+y_1^{403}+y_1^{402}+y_1^{401}+y_1^{399}+y_1^{397}+y_1^{395}+y_1^{394}+y_1^{392}+y_1^{388}+y_1^{387}+y_1^{386}+y_1^{384}+y_1^{380}+y_1^{378}+y_1^{377}+y_1^{
376}+y_1^{375}+y_1^{374}+y_1^{372}+y_1^{370}+y_1^{369}+y_1^{368}+y_1^{364}+y_1^{363}+y_1^{361}+y_1^{360}+y_1^{359}+y_1^{358}+y_1^{357}+y_1^{355}+y_1^{350}+y_1^{
347}+y_1^{345}+y_1^{343}+y_1^{340}+y_1^{338}+y_1^{336}+y_1^{334}+y_1^{330}+y_1^{329}+y_1^{327}+y_1^{326}+y_1^{325}+y_1^{324}+y_1^{321}+y_1^{319}+y_1^{318}+y_1^{
316}+y_1^{315}+y_1^{312}+y_1^{308}+y_1^{306}+y_1^{305}+y_1^{302}+y_1^{301}+y_1^{296}+y_1^{295}+y_1^{292}+y_1^{290}+y_1^{289}+y_1^{285}+y_1^{284}+y_1^{278}+y_1^{
277}+y_1^{276}+y_1^{275}+y_1^{274}+y_1^{273}+y_1^{272}+y_1^{271}+y_1^{265}+y_1^{261}+y_1^{260}+y_1^{256}+y_1^{255}+y_1^{250}+y_1^{249}+y_1^{248}+y_1^{247}+y_1^{
243}+y_1^{242}+y_1^{240}+y_1^{239}+y_1^{235}+y_1^{234}+y_1^{233}+y_1^{231}+y_1^{230}+y_1^{229}+y_1^{227}+y_1^{225}+y_1^{224}+y_1^{222}+y_1^{221}+y_1^{217}+y_1^{
215}+y_1^{213}+y_1^{212}+y_1^{210}+y_1^{209}+y_1^{207}+y_1^{205}+y_1^{203}+y_1^{202}+y_1^{201}+y_1^{200}+y_1^{199}+y_1^{197}+y_1^{195}+y_1^{189}+y_1^{187}+y_1^{
183}+y_1^{182}+y_1^{181}+y_1^{180}+y_1^{179}+y_1^{178}+y_1^{175}+y_1^{172}+y_1^{169}+y_1^{167}+y_1^{165}+y_1^{164}+y_1^{163}+y_1^{160}+y_1^{159}+y_1^{157}+y_1^{
155}+y_1^{154}+y_1^{145}+y_1^{141}+y_1^{137}+y_1^{133}+y_1^{130}+y_1^{129}+y_1^{128}+y_1^{125}+y_1^{123}+y_1^{122}+y_1^{121}+y_1^{117}+y_1^{115}+y_1^{114}+y_1^{
113}+y_1^{112}+y_1^{111}+y_1^{110}+y_1^{109}+y_1^{108}+y_1^{107}+y_1^{102}+y_1^{98}+y_1^{96}+y_1^{95}+y_1^{90}+y_1^{89}+y_1^{88}+y_1^{86}+y_1^{84}+y_1^{83}+y_1^{81}+y_1^{
80}+y_1^{78}+y_1^{77}+y_1^{76}+y_1^{74}+y_1^{72}+y_1^{70}+y_1^{68}+y_1^{67}+y_1^{65}+y_1^{63}+y_1^{62}+y_1^{61}+y_1^{55}+y_1^{54}+y_1^{53}+y_1^{52}+y_1^{51}+y_1^{50}+y_1^{49}+y_1^{
47}+y_1^{46}+y_1^{45}+y_1^{43}+y_1^{42}+y_1^{40}+y_1^{38}+y_1^{36}+y_1^{35}+y_1^{33}+y_1^{32}+y_1^{31}+y_1^{30}+y_1^{29}+y_1^{28}+y_1^{24}+y_1^{23}+y_1^{22}+y_1^{21}+y_1^{20}+y_1^{
17}+y_1^{15}+y_1^{14}+y_1^{13}+y_1^{11}+y_1^{12}+y_1^{9}+y_1^{7}+y_1^{6}+y_1^{4}+y_1^{3}+y_1^{2}+y_1^{1}+1+\\
{x_2}\cdot \bigl(y_1^{26}+y_1^{24}+y_1^{23}+y_1^{13}+y_1^{11}+y_1^{10}+y_1^{8}+y_1^{7}+y_1^{6}+y_1^{3}+y_1\bigr)\bigr)$}}\\ 
\hline

c & \parbox{15.5cm}{
     \scriptsize{${x_1^3}\cdot \bigl(y_1^{477}+y_1^{476}+y_1^{473}+y_1^{472}+y_1^{470}+y_1^{469}+y_1^{466}+y_1^{463}+y_1^{461}+y_1^{459}+y_1^{458}+y_1^{457}+y_1^{456}+y_1^{453}+y_1^{452}+
y_1^{451}+y_1^{450}+y_1^{449}+y_1^{448}+y_1^{447}+y_1^{446}+y_1^{443}+y_1^{441}+y_1^{440}+y_1^{439}+y_1^{438}+y_1^{436}+y_1^{433}+y_1^{431}+y_1^{428}+y_1^{422}+y_1^
{420}+y_1^{419}+y_1^{414}+y_1^{413}+y_1^{410}+y_1^{409}+y_1^{407}+y_1^{406}+y_1^{403}+y_1^{402}+y_1^{400}+y_1^{399}+y_1^{394}+y_1^{391}+y_1^{388}+y_1^{385}+y_1^{384}+y_1^{383}+y_1^{382}+y_1^{381}+y_1^{379}+y_1^{373}+y_1^{372}+y_1^{368}+y_1^{367}+y_1^{366}+y_1^{363}+y_1^{362}+y_1^{359}+y_1^{358}+y_1^{357}+y_1^{356}+y_1^{
354}+y_1^{353}+y_1^{350}+y_1^{349}+y_1^{348}+y_1^{347}+y_1^{344}+y_1^{342}+y_1^{341}+y_1^{340}+y_1^{339}+y_1^{337}+y_1^{335}+y_1^{334}+y_1^{333}+y_1^{332}+y_1^{
331}+y_1^{330}+y_1^{328}+y_1^{325}+y_1^{324}+y_1^{323}+y_1^{322}+y_1^{321}+y_1^{320}+y_1^{319}+y_1^{313}+y_1^{312}+y_1^{310}+y_1^{307}+y_1^{305}+y_1^{304}+y_1^{
303}+y_1^{302}+y_1^{300}+y_1^{295}+y_1^{294}+y_1^{293}+y_1^{292}+y_1^{289}+y_1^{287}+y_1^{286}+y_1^{283}+y_1^{282}+y_1^{280}+y_1^{279}+y_1^{276}+y_1^{274}+y_1^{
272}+y_1^{270}+y_1^{269}+y_1^{267}+y_1^{264}+y_1^{263}+y_1^{262}+y_1^{261}+y_1^{259}+y_1^{256}+y_1^{255}+y_1^{254}+y_1^{253}+y_1^{251}+y_1^{246}+y_1^{244}+y_1^{
243}+y_1^{242}+y_1^{241}+y_1^{238}+y_1^{237}+y_1^{235}+y_1^{234}+y_1^{233}+y_1^{231}+y_1^{230}+y_1^{225}+y_1^{222}+y_1^{221}+y_1^{220}+y_1^{212}+y_1^{210}+y_1^{
208}+y_1^{207}+y_1^{206}+y_1^{205}+y_1^{199}+y_1^{198}+y_1^{197}+y_1^{193}+y_1^{191}+y_1^{190}+y_1^{189}+y_1^{188}+y_1^{187}+y_1^{185}+y_1^{184}+y_1^{180}+y_1^{
179}+y_1^{177}+y_1^{176}+y_1^{175}+y_1^{174}+y_1^{170}+y_1^{169}+y_1^{167}+y_1^{166}+y_1^{165}+y_1^{162}+y_1^{160}+y_1^{159}+y_1^{158}+y_1^{156}+y_1^{155}+y_1^{
154}+y_1^{148}+y_1^{146}+y_1^{142}+y_1^{141}+y_1^{139}+y_1^{138}+y_1^{137}+y_1^{135}+y_1^{131}+y_1^{130}+y_1^{129}+y_1^{128}+y_1^{126}+y_1^{125}+y_1^{123}+y_1^{
122}+y_1^{120}+y_1^{117}+y_1^{115}+y_1^{112}+y_1^{111}+y_1^{110}+y_1^{108}+y_1^{107}+y_1^{105}+y_1^{103}+y_1^{102}+y_1^{101}+y_1^{99}+y_1^{97}+y_1^{96}+y_1^{92}+y_1^{
91}+y_1^{86}+y_1^{85}+y_1^{83}+y_1^{82}+y_1^{81}+y_1^{80}+y_1^{79}+y_1^{78}+y_1^{77}+y_1^{76}+y_1^{75}+y_1^{74}+y_1^{72}+y_1^{70}+y_1^{68}+y_1^{67}+y_1^{65}+y_1^{62}+y_1^{61}+y_1^{
59}+y_1^{58}+y_1^{57}+y_1^{56}+y_1^{55}+y_1^{54}+y_1^{53}+y_1^{52}+y_1^{51}+y_1^{50}+y_1^{49}+y_1^{48}+y_1^{45}+y_1^{44}+y_1^{43}+y_1^{41}+y_1^{40}+y_1^{39}+y_1^{38}+y_1^{37}+y_1^{
36}+y_1^{31}+y_1^{30}+y_1^{27}+y_1^{26}+y_1^{25}+y_1^{24}+y_1^{22}+y_1^{21}+y_1^{20}+y_1^{18}+y_1^{17}+y_1^{16}+y_1^{15}+y_1^{14}+y_1^{13}+y_1^{12}+y_1^{11}+y_1^{9}+y_1^{8}+y_1^{7}+y_1^{
6}+y_1^{5}+y_1^{4}+y_1^{3}+y_1^{2}+y_1+\\
{x_2}\cdot \bigr(y_1^{24}+y_1^{23}+y_1^{21}+y_1^{19}+y_1^{17}+y_1^{15}+y_1^{11}+y_1^{10}+y_1^{8}+y_1^{7}+y_1^{6}+y_1^{5}+y_1^{4}+y_1^{2}+1\bigl)\bigl) $}}\\
\hline

\end{tabular}
\label{table:n=55}
\end{center}
\end{table*}
\end{tiny}
\newpage
\section{Proof of equation \texorpdfstring{$\eqref{eq:iff51}$}{TEXT}} \label{sec:appendix}

\begin{proposition}
Let $C$ be the code with length $n=51$ defined by $S_C=\{0,1,5\}$. If $\mu=2$, then  
$$
(x_1^7+x_3)=0\quad \textrm{if and only if}\quad x_1^{51}=1.
$$
\end{proposition}
\begin{proof}

Let us suppose that $(x_1^7+x_3)=0$. Since we have that
\begin{equation}\label{eq:ax1^7}
x_1^7=(z_1+z_2)^7=z_1^7+z_2^7+(z_1 z_2)\left((z_1+z_2)^5+(z_1z_2)^2(z_1+z_2)\right),
\end{equation}
then $(x_1^7+x_3)=0$ and $\mu=2$ implie that $(z_1+z_2)^5=(z_1z_2)^2(z_1+z_2)$. So
\begin{equation}\label{eq:x255}
((z_1+z_2)^5)^{51}=((z_1z_2)^2(z_1+z_2))^{51}.
\end{equation} 
Since the splitting field of $x^{51}+1$ over $\FF_2$ is $\FF_{256}$, then $x_1^{255}=1$. 
But we have also that $z_1z_2\in\mathbb{F}_{256}$, so $(z_1z_2)^{2\times 51}=1$. Then, by \eqref{eq:x255} we get that $(z_1+z_2)^{51}=1$.

Vice-versa, if $x_1^{51}=1$ then $x_1^{64}=x_1^{13}$. But $x_1^{64}=z_1^{64}+z_2^{64}=z_1^{13}+z_2^{13}=x_{13}$, so $x_1^{13}=x_{13}$.
By Newton's identities \cite{CGC-cd-book-macwilliamsTOT}, we know that $x_1^5=x_5+b x_3$. So, 
\begin{equation}\label{eq:x13}
x_1^{13}=x_1^5x_1^8=(x_5+b x_3)x_1^8=x_5x_1^8+b x_3 x_1^8=x_{13}+b^5x_3+b x_3x_1^8.
\end{equation}
Since $x_1^{13}=x_{13}$, then, by \eqref{eq:x13}, $b^5 x_3+b x_1^8x_3=0$. So, since $b\neq 0$, either $b=x_1^2$ or $x_3=0$.\\
If $b=x_1^2$, then, by \eqref{eq:ax1^7}, $x_1^7=x_7+x_1^2(x_1^5+x_1^4x_1)=x_7$.\\
If $x_3=0$ then $x_1^5=x_5+b x_3=x_5$. Then 
$$x_1^7=x_1^5x_1^2=x_5x_1^2=x_7+b x_3$$.
Since $x_3=0$, then $x_1^7=x_7$.
\end{proof}

%%%%%%%%%%%%%%%%%%%%%%%%%%%%%%%%%%%%%%%%%%%%%%%%%%%%%%%%%%%%%%%%%%%%%
\section*{Acknowledgements}
Most results in this paper are from the last author's PHD thesis and so she would like
to thank her supervisor, the third author. The work of the second author  has been supported in part by EPSRC via grant EP/I03126X.

%%%%%%%%%%%%%%%%%%%%%%%%%%%%%%%%%%%%%%%%%%%%%%%%%%%%%%%%%%%%%%%%%%%%%
%\newpage
%\small
%\vspace{-0.37cm}

%\bibliographystyle{plain}
\bibliographystyle{ieeetr}
\bibliography{Refs}

%===================================================================

\end{document}